%% file: main.tex
\documentclass[a4paper,USenglish]{lipics-v2018}
\usepackage[utf8]{inputenc}
\usepackage[T1]{fontenc}

\usepackage{amsmath,amssymb,amsfonts,amsthm}
\usepackage{mathtools}
\usepackage{nicefrac}
\usepackage{csquotes}
\usepackage{xspace}
\usepackage{color}
\usepackage{graphicx}
\usepackage[labelfont=bf]{caption}
\usepackage{subcaption}
\usepackage[inline]{enumitem}
\usepackage{pdflscape}

\graphicspath{{../}}

\newtheorem{proposition}{Proposition}

\renewcommand{\paragraph}[1]{\par\textbf{#1}}

\makeatletter
\newcommand*{\rom}[1]{\expandafter\@slowromancap\romannumeral #1@}
\makeatother

\newcommand{\ie}[0]{i.e.\xspace}

\newcommand{\eg}[0]{e.g.\xspace}

\usepackage{algorithm}
\usepackage[noend]{algpseudocode}
\newcommand{\Algblankline}{\item[]}

\newcommand{\dist}[0]{\ensuremath{\delta}}
\newcommand{\abs}[1]{\left\vert #1 \right\vert}
\newcommand{\norm}[1]{\left\lVert #1 \right\rVert}

\newcommand{\Oh}[1]{\ensuremath{\mathcal{O}(#1)}}
\newcommand{\dF}{d_F}
\newcommand{\sigspatial}{\textsc{Sigspatial}\xspace}
\newcommand{\characters}{\textsc{Characters}\xspace}
\newcommand{\geolife}{\textsc{GeoLife}\xspace}
\newcommand{\giscup}{GIS Cup\xspace}
\newcommand{\HeurClose}{\textsc{HeurClose}}
\newcommand{\HeurFar}{\textsc{HeurFar}}

\DeclareMathOperator{\argmin}{arg\,min}

\title{Walking the Dog Fast in Practice: Algorithm Engineering of the Fréchet Distance}

\author{Karl Bringmann}{Max Planck Institute for Informatics, Saarland Informatics Campus, Saarbrücken, Germany}{kbringma@mpi-inf.mpg.de}{}{} %mandatory, please use full name; only 1 author per \author macro; first two parameters are mandatory, other parameters can be empty.

\author{Marvin Künnemann}{Max Planck Institute for Informatics, Saarland Informatics Campus, Saarbrücken, Germany}{marvin@mpi-inf.mpg.de}{}{} %mandatory, please use full name; only 1 author per \author macro; first two parameters are mandatory, other parameters can be empty.

\author{André Nusser}{Max Planck Institute for Informatics, Saarland Informatics Campus, Saarbrücken, Germany}{anusser@mpi-inf.mpg.de}{}{} %mandatory, please use full name; only 1 author per \author macro; first two parameters are mandatory, other parameters can be empty.

\authorrunning{K. Bringmann, M. Künnemann, and A. Nusser}%mandatory. First: Use abbreviated first/middle names. Second (only in severe cases): Use first author plus 'et al.'

\Copyright{Karl Bringmann, Marvin Künnemann, and André Nusser}%mandatory, please use full first names. LIPIcs license is "CC-BY";  http://creativecommons.org/licenses/by/3.0/

\subjclass{
,
}

%mandatory

\category{}%optional, e.g. invited paper

\relatedversion{}%optional, e.g. full version hosted on arXiv, HAL, or other respository/website

\supplement{\url{https://github.com/chaot4/frechet_distance}}%optional, e.g. related research data, source code, ... hosted on a repository like zenodo, figshare, GitHub, ...

\funding{}%optional, to capture a funding statement, which applies to all authors. Please enter author specific funding statements as fifth argument of the \author macro.

\EventEditors{John Q. Open and Joan R. Access}
\EventNoEds{2}
\EventLongTitle{42nd Conference on Very Important Topics (CVIT 2016)}
\EventShortTitle{CVIT 2016}
\EventAcronym{CVIT}
\EventYear{2016}
\EventDate{December 24--27, 2016}
\EventLocation{Little Whinging, United Kingdom}
\EventLogo{}
\SeriesVolume{42}
\ArticleNo{p}
\nolinenumbers
\hideLIPIcs  %uncomment to remove references to LIPIcs series (logo, DOI, ...), e.g. when preparing a pre-final version to be uploaded to arXiv or another public repository
\begin{document}

\maketitle
\input{trunk/abstract.tex}

\input{trunk/introduction.tex}
\input{trunk/preliminaries.tex}
\input{trunk/complete_decider.tex}
\input{trunk/decider.tex}
\input{trunk/query.tex}
\input{trunk/implementation_details.tex}
\input{trunk/experiments.tex}
\input{trunk/certificates.tex}
\input{trunk/conclusion.tex}

\bibliographystyle{plainurl}% the recommnded bibstyle
\bibliography{trunk/biblio}

\end{document}

%% file: trunk/abstract.tex
\begin{abstract}
\noindent
The Fréchet distance provides a natural and intuitive measure for the popular task of computing the similarity of two (polygonal) curves.
While a simple algorithm computes it in near-quadratic time, a strongly subquadratic algorithm cannot exist unless the Strong Exponential Time Hypothesis fails.
Still, fast practical implementations of the Fréchet distance, in particular for \emph{realistic} input curves, are highly desirable.
This has even lead to a designated competition, the ACM SIGSPATIAL GIS Cup 2017: Here, the challenge was to implement a near-neighbor data structure under the Fréchet distance.
The bottleneck of the top three implementations turned out to be precisely the decision procedure for the Fréchet distance.

In this work, we present a fast, certifying implementation for deciding the Fréchet distance, in order to (1) complement its pessimistic worst-case hardness by an empirical analysis on realistic input data and to (2) improve the state of the art for the GIS Cup challenge.
We experimentally evaluate our implementation on a large benchmark consisting of several data sets (including handwritten characters and GPS trajectories). Compared to the winning implementation of the GIS Cup, we obtain running time improvements of up to more than two orders of magnitude for the decision procedure and of up to a factor of 30 for queries to the near-neighbor data structure.
\end{abstract}

%% file: trunk/introduction.tex
%!TEX root = ../main.tex

\section{Introduction}\label{sec:intro}

A variety of practical applications analyze and process trajectory data coming from different sources like GPS measurements, digitized handwriting, motion capturing, and many more. One elementary task on trajectories is to compare them, for example in the context of signature verification \cite{signatureverification}, map matching \cite{chambers2018map, wylie2014intermittent, chen2011approximate, buchin2009exact}, and clustering \cite{buchin2018, campbell2015clustering}. In this work we consider the Fréchet distance as curve similarity measure as it is arguably the most natural and popular one.
Intuitively, the Fréchet distance between two curves is explained using the following analogy. A person walks a dog, connected by a leash.
%Both have a fixed trajectory, however, there are possibilities to vary the traversal as long as it starts at the first node, ends at the last node, is continuous and monotone.
Both walk along their respective curve, with possibly varying speeds and without ever walking backwards.
Over all such traversals, we search for the ones which minimize the leash length, \ie, we minimize the maximal distance the dog and the person have during the traversal. 

Initially defined more than one hundred years ago \cite{frechet1906}, the Fréchet distance quickly gained popularity in computer science after the first algorithm to compute it was presented by Alt and Godau~\cite{alt_95}. In particular, they showed how to decide whether two length-$n$ curves have Fréchet distance at most $\delta$ in time $\Oh{n^2}$ by full exploration of a quadratic-sized search space, the so-called \emph{free-space} (we refer to Section~\ref{subsec:freespace_diagram} for a definition). Almost twenty years later, it was shown that, conditional on the Strong Exponential Time Hypothesis (SETH), there cannot exist an algorithm with running time $\Oh{n^{2-\epsilon}}$ for any $\epsilon > 0$ \cite{bringmann2014}.
Even for realistic models of input curves, such as $c$-packed curves \cite{driemel2012}, exact computation of the Fréchet distance requires time $n^{2-o(1)}$ under SETH~\cite{bringmann2014}. Only if we relax the goal to finding a $(1+\epsilon)$-approximation of the Fréchet distance, algorithms with near-linear running times in $n$ and $c$ on $c$-packed curves are known to exist~\cite{driemel2012, bk2017}.

% There is still the possibility to search for better algorithms on certain curve models -- even though the conditional lower bound was also shown for practical curve models such as $c$-packed curves \cite{bringmann2014, driemel2012}.
% TODO: even realistic curve models don't have better algs, however, better approximation algorithms.

It is a natural question whether these hardness results are mere theoretical worst-case results or whether computing the Fréchet distance is also hard in practice.
This line of research was particularly fostered by the research community in form of the  GIS Cup 2017~\cite{giscup17}.
In this competition, the 28 contesting teams  were challenged to give a fast implementation for the following problem: Given a data set of two-dimensional trajectories $\mathcal{D}$, answer queries that ask to return, given a curve $\pi$ and query distance $\delta$, all $\sigma \in \mathcal{D}$ with Fréchet distance at most $\delta$ to $\pi$. We call this the \emph{near-neighbor problem}.
% of the form $(\pi, \delta)$, where all curves $\sigma \in \mathcal{D}$ should be returned for which the Fréchet distance to $\pi$ is less than $\delta$.

The three top implementations \cite{sigspatial1, sigspatial2, sigspatial3} use multiple layers of heuristic filters and spatial hashing to decide as early as possible whether a curve belongs to the output set or not, and finally use an essentially exhaustive Fréchet distance computation for the remaining cases. Specifically, these implementations perform the following steps:
\begin{enumerate}[label=(\arabic*), ref=\arabic*]
	\setcounter{enumi}{-1}
	\item Preprocess $\mathcal{D}$. \label{enum:preprocessing}
\end{enumerate}
	On receiving a query with curve $\pi$ and query distance $\delta$:
\begin{enumerate}[label=(\arabic*), ref=\arabic*]
	\item Use spatial hashing to identify candidate curves  $\sigma \in \mathcal{D}$.\label{enum:spatialHashing}
	\item For each candidate $\sigma$, decide whether $\pi,\sigma$ have Fréchet distance $\le\delta$:\label{enum:decider}
	\begin{enumerate}[label=\alph*), ref=\alph*]
	\item Use heuristics (\emph{filters}) for a quick resolution in simple cases. \label{enum:filter}
	\item If unsuccessful, use a \emph{complete decision procedure via free-space exploration}.\label{enum:complete}
	\end{enumerate}
\end{enumerate}
Let us highlight the \emph{Fréchet decider} outlined in steps \ref{enum:decider}\ref{enum:filter} and \ref{enum:decider}\ref{enum:complete}: %that returns whether $\dF(\pi, \sigma) \le \delta$, where $\dF(\pi, \sigma)$ denotes the Fréchet distance of $\pi, \sigma$. 
Here, \emph{filters} refer to sound, but incomplete Fréchet distance decision procedures, i.e., whenever they succeed to find an answer, they are correct, but they may return that the answer remains unknown. In contrast, a \emph{complete decision procedure via free-space exploration} explores a sufficient part of the free space (the search space) to always determine the correct answer. %(i.e., they may return any of the outputs ``$\dF(\pi, \sigma) \le \delta$'', ``$\dF(\pi, \sigma) > \delta$'' or ``unknown'', where the output in the first two cases must always be correct). 
As it turns out, the bottleneck in all three implementations is precisely Step~\ref{enum:decider}\ref{enum:complete}, the complete decision procedure via free-space exploration.
Especially \cite{sigspatial1} improved upon the naive implementation of the free-space exploration by designing very basic pruning rules, which might be the advantage because of which they won the competition.
There are two directions for further substantial improvements over the cup implementations: (1) increasing the range of instances covered by fast filters and (2) algorithmic improvements of the exploration of the reachable free-space.
% \marvin{always write \emph{filters} and \emph{complete decision procedure via free-space exploration}, i.e., use them as special highlighted terms?}

\paragraph{Our Contribution.}
% To obtain a state-of-the art implementation of a Fréchet distance decider, we improve Step~\ref{enum:decider}\ref{enum:complete}. As it turns out, our improvements will subsequently also motivate a modification of Step~\ref{enum:decider}\ref{enum:filter}, i.e., speeding up known filters, to ensure that these filters are not outperformed by our improved complete decision procedure via free-space exploration.
%We continue along this line of research and our main contributions are thus the engineering efforts that we put into improving the practical runtime of the exact decider.
We develop a fast, practical Fréchet distance implementation. To this end, we give a complete decision procedure via free-space exploration that uses a divide-and-conquer interpretation of the Alt-Godau algorithm for the Fréchet distance and optimize it using sophisticated pruning rules. These pruning rules greatly reduce the search space for the realistic benchmark sets we consider -- this is surprising given that simple constructions generate hard instances which require the exploration of essentially the full quadratic-sized search space \cite{bringmann2014, bringmannMulzer2015}. Furthermore, we present improved filters that are sufficiently fast compared to the complete decider. Here, the idea is to use adaptive step sizes (combined with useful heuristic tests) to achieve essentially ``sublinear'' time behavior for testing if an instance can be resolved quickly. Additionally, our implementation is certifying (see~\cite{McConnellMNS11} for a survey on certifying algorithms), meaning that for every decision of curves being close/far, we provide a short proof (certificate) that can be checked easily; we also implemented a computational check of these certificates. See Section \ref{sec:certificates} for details.
% Also, the source code and the benchmarks should be publicly available to make it possible for new approaches to be tested against ours. All of the above mentioned properties are fulfilled for this paper. Additionally, we particularly focused on making our implementation easily readable to also enable others to understand the details and encourage them to adapt and reuse the code.

An additional contribution of this work is the creation of benchmarks to make future implementations more easily comparable. We compile benchmarks both for the near-neighbor problem (Steps \ref{enum:preprocessing} to \ref{enum:decider}) and for the decision problem (Step \ref{enum:decider}). For this, we used publicly available curve data and created queries in a way that should be representative for the performance analysis of an implementation. As data sets we use the GIS Cup trajectories~\cite{sigspatial_dataset}, a set of handwritten characters called the Character Trajectories Data Set~\cite{characters_dataset} from \cite{Dua:2017}, and the GeoLife data set~\cite{geolife_dataset} of Microsoft Research \cite{geolife1, geolife2, geolife3}. Our benchmarks cover different distances and also curves of different similarity, giving a broad overview over different settings.
%The query benchmark tests the performance of queries which differ in the size of the result, which also gives a good overview over the performance on realistic data.
We make the source code as well as the benchmarks publicly available to enable independent comparisons with our approach.\footnote{Code and benchmarks are available at: \url{https://github.com/chaot4/frechet_distance}} Additionally, we particularly focus on making our implementation easily readable to enable and encourage others to reuse the code.

\paragraph{Evaluation.}
The GIS Cup 2017 had 28 submissions, with the top three submissions\footnote{The submissions were evaluated \enquote{for their correctness and average performance on a[sic!] various large trajectory databases and queries}. Additional criteria were the following: \enquote{We will use the total elapsed wall clock time as a measure of performance. For breaking ties, we will first look into the scalability behavior for more and more queries on larger and larger datasets. Finally, we break ties on code stability, quality, and readability and by using different datasets.}} (in decreasing order) due to Bringmann and Baldus~\cite{sigspatial1}, Buchin et al.~\cite{sigspatial2}, and Dütsch and Vahrenhold~\cite{sigspatial3}. We compare our implementation with all of them by running their implementations on our new benchmark set for the near-neighbor problem and also comparing to the improved decider of \cite{sigspatial1}. The comparison shows significant speed-ups up to almost a factor of 30 for the near-neighbor problem and up to more than two orders of magnitude for the decider.
% Naturally, not all details of our implementation can be included in the paper -- for a full coverage of the details, we refer to the source code.

\paragraph{Related Work.}
The best known algorithm for deciding the Fréchet distance runs in time $O(n^2 \frac{(\log \log n)^{2}}{\log n})$ on the word RAM~\cite{BuchinBMM17}.
This relies on the Four Russians technique and is mostly of theoretical interest.
% For the discrete variant of the Fréchet distance, a slightly faster algorithm running in time $O(n^2 \frac{\log \log n}{\log n})$ is known~\cite{AgarwalBKS14}.
There are many variants of the Fréchet distance, \eg, the discrete Fréchet distance~\cite{AgarwalBKS14, EiterM94}.
After the GIS Cup 2017, several practical papers studying aspects of the Fréchet distance appeared~\cite{ACKGS18, CeccarelloDS18,Wei2018}. Some of this work~\cite{ACKGS18, CeccarelloDS18} addressed how to improve upon the spatial hashing step (Step~\ref{enum:spatialHashing}) if we relax the requirement of exactness. Since this is orthogonal to our approach of improving the complete decider, these improvements could possibly be combined with our algorithm.
The other work~\cite{Wei2018} neither compared with the GIS Cup implementations, nor provided their source code publicly to allow for a comparison, which is why we have to ignore it here.

\paragraph{Organization.} First, in Section~\ref{sec:prelims}, we present all the core definitions. Subsequently, we explain our complete decider in Section \ref{sec:complete_decider}. The following section then explains the decider and its filtering steps. Then, in Section \ref{sec:query}, we present a query data structure which enables us to compare to the GIS Cup submissions. Section~\ref{sec:implementation_details} contains some details regarding the implementation to highlight crucial points we deem relevant for similar implementations. We conduct extensive experiments in Section \ref{sec:experiments}, detailing the improvements over the current state of the art by our implementation. Finally, in Section \ref{sec:certificates}, we describe how we make our implementation certifying and evaluate the certifying code experimentally.

%% file: trunk/preliminaries.tex
\section{Preliminaries} \label{sec:prelims}
Our implementation as well as the description are restricted to two dimensions, however, the approach can also be generalized to polygonal curves in $d$ dimensions.
Therefore, a \emph{curve} $\pi$ is defined by its \emph{vertices} $\pi_1, \dots, \pi_n \in \mathbb{R}^2$ which are connected by straight lines.
We also allow continuous indices as follows. For $p = i + \lambda$ with $i \in \{1, \dots, n\}$ and $\lambda \in [0,1]$, let
\[
	\pi_p \coloneqq (1-\lambda) \pi_i + \lambda \pi_{i+1}.
\]
We call the $\pi_p$ with $p \in [1,n]$ the \emph{points} on $\pi$. A subcurve of $\pi$ which starts at point $p$ and ends at point $q$ on $\pi$ is denoted by $\pi_{p \dots q}$. In the remainder, we denote the \emph{number of vertices} of $\pi$ (resp.\ $\sigma$) with $n$ (resp.\ $m$) if not stated otherwise. We denote the length of a curve $\pi$ by $\norm{\pi}$, \ie, the sum of the Euclidean lengths of its line segments. Additionally, we use $\norm{v}$ for the Euclidean norm of a vector $v \in \mathbb{R}^2$.
For two curves $\pi$ and $\sigma$, the \emph{Fréchet distance} $d_F(\pi, \sigma)$ is defined as
\[
	d_F(\pi, \sigma) \coloneqq \inf\limits_{\substack{f \in \mathcal{T}_n\\g \in \mathcal{T}_m}} \max_{t \in [0,1]} \norm{\pi_{f(t)} - \sigma_{g(t)}},
\]
where $\mathcal{T}_k$ is the set of monotone and continuous functions $f: [0,1] \to [1,k]$. We define a \emph{traversal} as a pair $(f,g) \in \mathcal{T}_n \times \mathcal{T}_m$. Given two curves $\pi, \sigma$ and a query distance $\dist$, we call them \emph{close} if $d_F(\pi, \sigma) \leq \delta$ and \emph{far} otherwise.
There are two problem settings that we consider in this paper:
% \begin{enumerate}[label=(\arabic*)]
%     \item Deciding the Fréchet distance for two curves
%     \item Querying a set of curves for close curves
% \end{enumerate}
% These settings are explained in the remainder of this section.

\begin{description}
	\item[Decider Setting:]
	Given curves $\pi, \sigma$ and a distance $\delta$, decide whether $d_F(\pi, \sigma) \leq \delta$.
	(With such a decider, we can compute the exact distance by using parametric search in theory and binary search in practice.)

	\item[Query Setting:]
	Given a curve dataset $\mathcal{D}$, build a data structure that on query $(\pi, \delta)$ returns all $\sigma \in \mathcal{D}$ with $d_F(\pi, \sigma) \leq \delta$.
\end{description}
We mainly focus on the decider in this work. To allow for a comparison with previous implementations (which are all in the query setting), we also run experiments with our decider plugged into a data structure for the query setting.

\subsection{Preprocessing} \label{sec:preprocessing}
When reading the input curves we immediately compute additional data which is stored with each curve:
\begin{description}
	\item[Prefix Distances:] To be able to quickly compute the curve length between any two vertices of $\pi$, we precompute the prefix lengths, \ie, the curve lengths $\norm{\pi_{1 \dots i}}$ for every $i \in \{2, \dots, n\}$. We can then compute the curve length for two indices $i < i'$ on $\pi$ by $\norm{\pi_{i \dots i'}} = \norm{\pi_{1 \dots i'}} - \norm{\pi_{1 \dots i}}$.
	\item[Bounding Box:] We compute the bounding box of all curves, which is a simple coordinate-wise maximum and minimum computation.
\end{description}
Both of these preprocessing steps are extremely cheap as they only require a single pass over all curves, which we anyway do when parsing them. In the remainder of this work we assume that this additional data was already computed, in particular, we do not measure it in our experiments as it is dominated by reading the curves.

%% file: trunk/complete_decider.tex
\section{Complete Decider} \label{sec:complete_decider}

The key improvement of this work lies in the complete decider via free-space exploration. Here, we use a divide-and-conquer interpretation of the algorithm of Alt and Godau~\cite{alt_95} which is similar to \cite{sigspatial1} where a free-space diagram is built recursively. This interpretation allows us to prune away large parts of the search space by designing powerful \emph{pruning rules} identifying parts of the search space that are irrelevant for determining the correct output. Before describing the details, we formally define the free-space diagram.

\subsection{Free-Space Diagram} \label{subsec:freespace_diagram}
The free-space diagram was first defined in \cite{alt_95}. Given two polygonal curves $\pi$ and $\sigma$ and a distance \dist, it is defined as the set of all pairs of indices of points from $\pi$ and $\sigma$ that are close to each other, \ie,
\[
	F \coloneqq \{ (p,q) \in [1,n] \times [1,m] \mid \norm{\pi_p - \sigma_q} \leq \dist \}.
\]
For an example see Figure \ref{fig:free-space_diagram}.
A \emph{path} from $a$ to $b$ in the free-space diagram $F$ is defined as a continuous mapping $P: [0,1] \to F$ with $P(0) = a$ and $P(1) = b$.
A path $P$ in the free-space diagram is \emph{monotone} if $P(x)$ is component-wise at most $P(y)$ for any $0 \leq x \leq y \leq 1$. The \emph{reachable space} is then defined as
\[
	R \coloneqq \{ (p,q) \in F \mid \text{there exists a monotone path from $(1,1)$ to $(p,q)$ in $F$}\}.
\]
Figure \ref{fig:reachable_free-space_diagram} shows the reachable space for the free-space diagram of Figure \ref{fig:free-space_diagram}. It is well known that $d_F(\pi, \sigma) \leq \dist$ if and only if $(n,m) \in R$.
% A monotone path in $F$ can be transformed into a monotone, continuous traversal of $\pi$ and $\sigma$. Thus, the set $R$ contains exactly those pairs of indices $(p,q)$ for which $d_F(\pi_{1 \dots p}, \sigma_{1 \dots q}) \leq \dist$. It follows that $d_F(\pi, \sigma) \leq \dist$ iff $(n,m) \in R$.

\begin{figure}
	\centering
	\includegraphics[width=\textwidth]{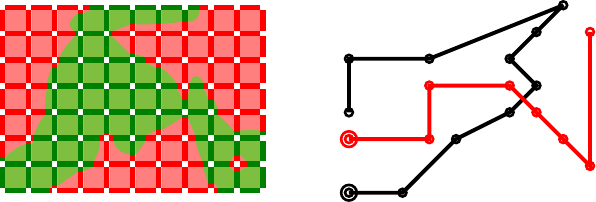}
	\caption{Example of a free-space diagram for curves $\pi$ (black) and $\sigma$ (red). The doubly-circled vertices mark the start. The free-space, \ie, the pairs of indices of points which are close, is colored green. The non-free areas are colored red. The threshold distance $\delta$ is roughly the distance between the first vertex of $\sigma$ and the third vertex of $\pi$.}
	\label{fig:free-space_diagram}
\end{figure}
\begin{figure}
	\centering
	\includegraphics[width=.5\textwidth]{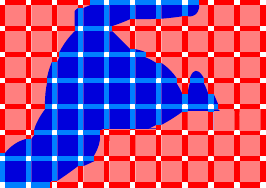}
	\caption{Reachable space of the free-space diagram in Figure \ref{fig:free-space_diagram}. The reachable part is blue and the non-reachable part is red. Note that the reachable part is a subset of the free-space.}
	\label{fig:reachable_free-space_diagram}
\end{figure}

This leads us to a simple dynamic programming algorithm to decide whether the Fréchet distance of two curves is at most some threshold distance. We iteratively compute $R$ starting from $(1,1)$ and ending at $(n,m)$, using the previously computed values. As $R$ is potentially a set of infinite size, we have to discretize it. A natural choice is to restrict to cells. The \emph{cell} of $R$ with coordinates $(i,j) \in \{1, \dots, n-1\} \times \{1, \dots, m-1\}$ is defined as $C_{i,j} \coloneqq [i,i+1] \times [j,j+1]$. This is a natural choice as given $R \cap C_{i-1,j}$ and $R \cap C_{i,j-1}$, we can compute $R \cap C_{i,j}$ in constant time; this follows from the simple fact that $F \cap C_{i,j}$ is convex~\cite{alt_95}. We call this computation of the outputs of a cell the \emph{cell propagation}. 
This algorithm runs in time $\mathcal{O}(nm)$ and was introduced by Alt and Godau~\cite{alt_95}.

\subsection{Basic Algorithm}
For integers $1 \leq i \leq i' \leq n, 1 \leq j \leq j' \leq m$ we call the set $B = [i,i'] \times [j,j']$ a \emph{box}. We denote the left/right/bottom/top \emph{boundaries} of $B$ by $B_l = \{i\} \times [j,j'], B_r = \{i'\} \times [j,j'], B_b = [i,i'] \times \{j\}, B_t = [i,i'] \times \{j'\}$. The \emph{left input} of $B$ is $B_l^R = B_l \cap R$, and its \emph{bottom input} is $B_b^R = B_b \cap R$. Similarly, the \emph{right/top output} of $B$ is $B_r^R = B_r \cap R$, $B_t^R = B_t \cap R$. A box is a cell if $i+1=i'$ and $j+1 = j'$.
We always denote the lower left corner of a box by $(i,j)$ and the top right by $(i',j')$, if not mentioned otherwise.

A recursive variant of the standard free-space decision procedure is as follows: Start with $B = [1,n] \times [1,m]$. At any recursive call, if $B$ is a cell, then determine its outputs from its inputs in constant time, as described by \cite{alt_95}. Otherwise, split $B$ vertically or horizontally into $B_1, B_2$ and first compute the outputs of $B_1$ from the inputs of $B$ and then compute the outputs of $B_2$ from the inputs of $B$ and the outputs of $B_1$. In the end, we just have to check $(n,m) \in R$ to decide whether the curves are close or far. This is a constant-time operation after calculating all outputs.

Now comes the main idea of our approach: we try to avoid recursive splitting by directly computing the outputs for non-cell boxes using certain rules. We call them \emph{pruning rules} as they enable pruning large parts of the recursion tree induced by the divide-and-conquer approach. Our pruning rules are heuristic, meaning that they are not always applicable, however, we show in the experiments that on practical curves they apply very often and therefore massively reduce the number of recursive calls.
The detailed pruning rules are described Section \ref{sec:pruning_rules}. 
Using these rules, we change the above recursive algorithm as follows. In any recursive call on box $B$, we first try to apply the pruning rules. If this is successful, then we obtained the outputs of $B$ and we are done with this recursive call. Otherwise, we perform the usual recursive splitting. Corresponding pseudocode is shown in Algorithm~\ref{alg:recursive_decider}.

\begin{algorithm}[t]
\caption{Recursive Decider of the Fréchet Distance}
\begin{algorithmic}[1]
\Procedure{DecideFréchetDistance}{$\pi, \sigma$}
\State \Call{ComputeOutputs}{$\pi, \sigma, [1,n] \times [1,m]$}
\State \Return $[(n,m) \in R]$
\EndProcedure
\Algblankline
\Procedure{ComputeOutputs}{$\pi, \sigma, B = [i,i'] \times [j,j']$}
\If{$B$ is a cell}
\State compute outputs by cell propagation
\Else
\State use pruning rules \rom{1} to \rom{4} to compute outputs of $B$
	\If{not all outputs have been computed}
		\If{$j' - j > i' - i$} \Comment{split horizontally}
			\State $B_1 = [i,i'] \times [j, \lfloor (j+j')/2 \rfloor]$
			\State $B_2 = [i,i'] \times [\lfloor (j+j')/2 \rfloor, j']$
		\Else \Comment{split vertically}
			\State $B_1 = [i, \lfloor (i+i')/2 \rfloor] \times [j,j']$
			\State $B_2 = [\lfloor (i+i')/2 \rfloor, i'] \times [j,j']$
		\EndIf
		\State \Call{ComputeOutputs}{$\pi, \sigma, B_1$}
		\State \Call{ComputeOutputs}{$\pi, \sigma, B_2$}
	\EndIf
\EndIf
\EndProcedure
\end{algorithmic}
\label{alg:recursive_decider}
\end{algorithm}

In the remainder of this section, we describe our pruning rules and their effects.
%As already mentioned, the core part of the complete decider is the computation of outputs of non-cell blocks and we describe the methods we employ and their effects in the remainder of this section.

\subsection{Pruning Rules} \label{sec:pruning_rules}
In this section we introduce the rules that we use to compute outputs of boxes which are above cell-level in certain special cases.
% To formulate the rules, we need to define different types of boundaries. We use \emph{empty} for boundaries on which no point is reachable. Later, for Rule \rom{3}, we also introduce and use the notion of a boundary being \emph{simple}.
Note that we aim at catching special cases which occur often in practice, as we cannot hope for improvements on adversarial instances due to the conditional lower bound of \cite{bringmann2014}. Therefore, we make no claims whether they are applicable, only that they are sound and fast. In what follows, we call a boundary \emph{empty} if its intersection with $R$ is $\emptyset$.

\subsubsection*{Rule \rom{1}: Empty Inputs}
The simplest case where we can compute the outputs of a box $B$ is if both inputs are empty, \ie $B_b^R = B_l^R = \emptyset$. In this case no propagation of reachability is possible and thus the outputs are empty as well, \ie $B_t^R = B_r^R = \emptyset$. See Figure \ref{fig:empty_inputs} for an example.

\begin{figure}
	\includegraphics[width=\textwidth]{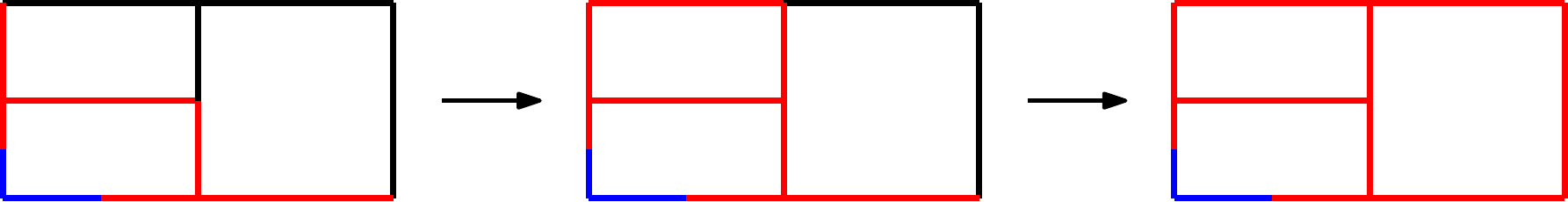}
	\caption{Output computation of a box when inputs are empty. First we can compute the outputs of the top left box and then the outputs of the right box. In this example, we then know that the curves have a Fréchet distance greater than $\delta$ as $(n,m)$ is not reachable.}
	\label{fig:empty_inputs}
\end{figure}

\subsubsection*{Rule \rom{2}: Shrink Box}
Instead of directly computing the outputs, this rule allows us to shrink the box we are currently working on, which reduces the problem size. Assume that for a box $B$ we have that $B_b^R = \emptyset$ and the lowest point of $B_l^R$ is $(i,j_{\min})$ with $j_{\min} > j$. In this case, no pair in $[i,i'] \times [j,j_{\min}]$ is reachable. Thus, we can shrink the box to the coordinates $[i,i'] \times [\lfloor j_{\min} \rfloor,j']$ without losing any reachability information. An equivalent rule can be applied if we swap the role of $B_b$ and $B_l$. See Figure~\ref{fig:shrink_box} for an example of applying this rule.

\begin{figure}
	\includegraphics[width=\textwidth]{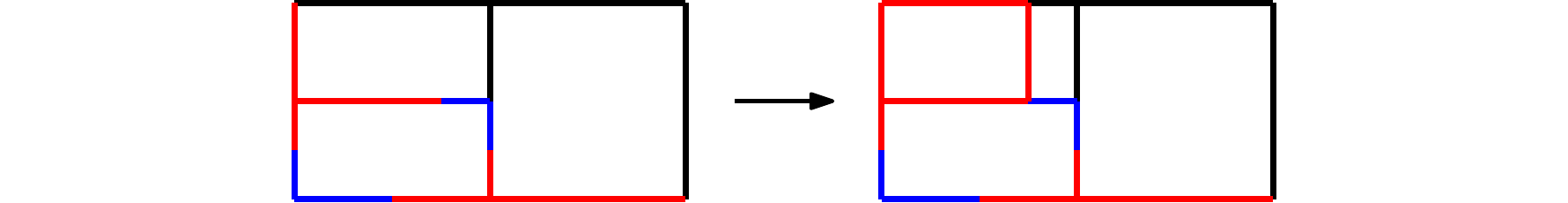}
	\caption{This is an example of shrinking a box in case one of the inputs is empty and the other one starts with an empty part. In this example the top left box has an empty input on the left and the start of the bottom input is empty as well. Thus, we can shrink the box to the right part.}
	\label{fig:shrink_box}
\end{figure}

\subsubsection*{Rule \rom{3}: Simple Boundaries}
\emph{Simple boundaries} are boundaries of a box that contain at most one free section.
To define this formally, a set $\mathcal{I} \subseteq [1,n] \times [1,m]$ is called an \emph{interval} if $\mathcal{I} = \emptyset$ or $\mathcal{I} = \{p\} \times [q,q']$ or $\mathcal{I} = [q,q'] \times \{p\}$ for real $p$ and an interval $[q,q']$. In particular, the four boundaries of a box $B = [i,i'] \times [j,j']$ are intervals. We say that an interval $\mathcal{I}$ is \emph{simple} if $\mathcal{I} \cap F$ is again an interval.
Geometrically, we have a free interval of a point $\pi_p$ and a curve $\sigma_{q \dots q'}$ (which is the form of a boundary in the free-space diagram) if the circle of radius \dist{} around $\pi_p$ intersects $\sigma_{q \dots q'}$ at most twice. See Figure \ref{fig:simple_interval} for an example. We call such a boundary simple because it is of low complexity, which we can exploit for pruning.

\begin{figure}
	\centering
	\includegraphics[width=.5\textwidth]{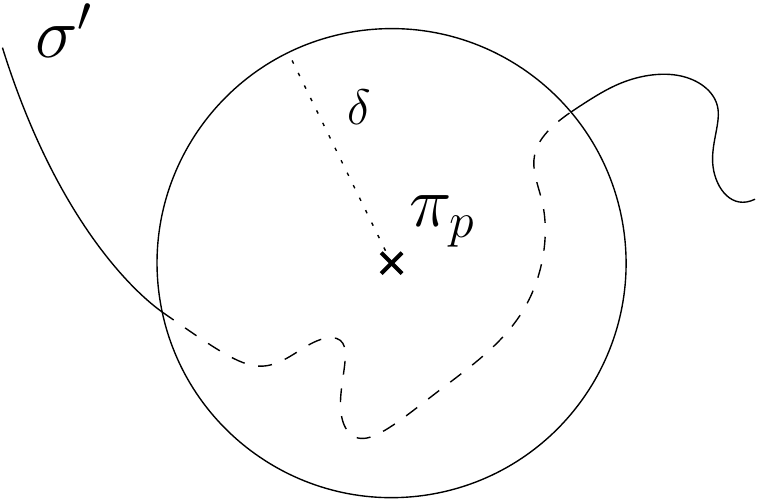}
	\caption{Example of a point $\pi_p$ and a curve $\sigma'$ which lead to a simple boundary.}
	\label{fig:simple_interval}
\end{figure}

There are three pruning rules that we do based on simple boundaries (see Figure \ref{fig:simple_interval_rules} for visualizations). They are stated here for the top boundary $B_t$, but symmetric rules apply to $B_r$. Later, in Section \ref{subsec:simple_boundaries_impl}, we then explain how to actually compute simple boundaries, \ie, also how to compute $B_t \cap F$. The pruning rules are:
\begin{enumerate}[label=(\alph*)]
	\item If $B_t$ is simple because $B_t \cap F$ is empty then we also know that the output of this boundary is empty. Thus, we are done with $B_t$.
	\item Suppose that $B_t$ is simple and, more specifically, of the form that it first has a free and then a non-free part; in other words, we have $(i,j') \in B_t \cap F$. Due to our recursive approach, we already computed the left inputs of the box and thus know whether the top left corner of the box is reachable, \ie whether $(i,j') \in R$.
		If this is the case, then we also know the reachable part of our simple boundary: Since $(i,j') \in R$ and $B_t \cap F$ is an interval containing $(i,j')$, we conclude that $B_t^R = B_t \cap F$ and we are done with $B_t$.
%If this is the case then we also know the reachable part of our simple boundary: it is exactly the free interval of the simple boundary, as everything which is free is reachable already from $(i, j')$ and non-free parts can never be reachable.
	\item
	% Again assume that the top boundary of the box is simple. However, this time the boundary starts with a non-free part but still contains a free interval. If we can determine somehow that the first point of the free interval -- we call it $p \coloneqq (i_p,j')$ in the following -- is reachable, then we again know the outputs of the complete boundary, because the reachability propagates to the whole free interval and everything non-free again cannot become reachable.
	Suppose that $B_t$ is simple, but the leftmost point $(i_{\min}, j')$ of $B_t \cap F$ has $i_{\min} > i$. In this case, we try to certify that $(i_{\min},j') \in R$, because then it follows that $B_t^R = B_t \cap F$ and we are done with $B_t$.
	To check for reachability of $(i_{\min},j')$, we try to propagate the reachability through the inside of the box, which in this case means to propagate it from the bottom boundary.
	%It is necessary that $(i_p, j)$ is reachable on the bottom boundary, \ie that $(i_p, j) \in R$. If this is the case then we check if $i_p \times [j,j']$ is completely free\footnote{We can compute this the same way as we compute simple boundaries, which is explained in Section \ref{subsec:simple_boundaries_impl}.}. In this case we can propagate the reachability through the box and thus can also decide the output.
	We test whether $(i_{\min},j)$ is in the input, \ie, if $(i_{\min},j) \in B_b^R$, and whether $\{i_{\min}\} \times [j,j'] \subseteq F$ (by slightly modifying the algorithm for simple boundary computations). If this is the case, then we can reach every point in $B_t \cap F$ from $(i_{\min},j)$ via $\{i_{\min}\} \times [j,j']$.
	Note that this is an operation in the complete decider where we explicitly use the \emph{inside} of a box and not exclusively operate on its boundaries.
\end{enumerate}
We also use symmetric rules by swapping \enquote{top} with \enquote{right} and \enquote{bottom} with \enquote{left}.

\begin{figure}
	\includegraphics[width=\textwidth]{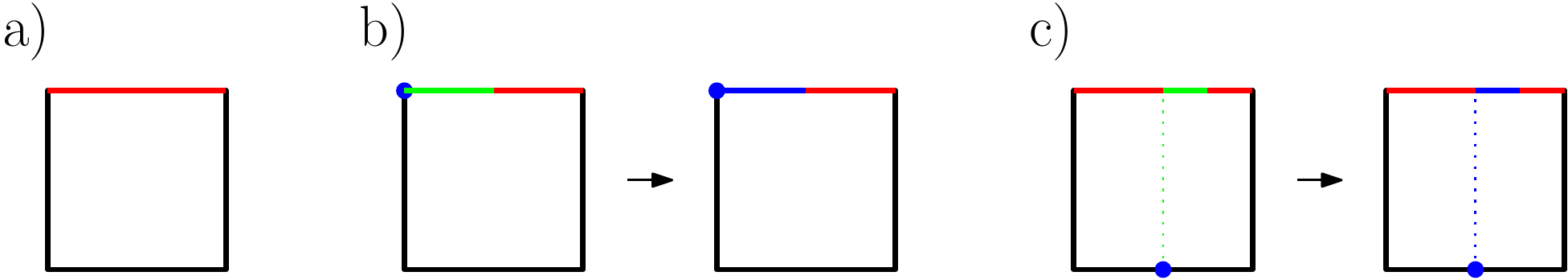}
	\caption{Visualization of the rules for computing outputs using simple boundaries. All three cases are visualized with the top boundary being simple. In a) the boundary is non-free and therefore no point on it can be reachable. In b) the boundary's beginning is free and reachable, enabling us to propagate the reachability to the entire free interval. In c) we can propagate the reachability of a point on the bottom boundary, using a free interval inside the box, to the beginning of the free interval of the top boundary and thus decide the entire boundary. The rules for the right boundary being simple are equivalent.}
	\label{fig:simple_interval_rules}
\end{figure}

\subsubsection*{Rule \rom{4}: Boxes at Free-Space Diagram Boundaries}
The boundaries of a free-space diagram are a special form of boundary which allows us to introduce an additional rule. Consider a box $B$ which touches the top boundary of the free-space diagram, \ie, $B = [i,i'] \times [j,m]$. Suppose the previous rules allowed us to determine the output for $B_r^R$. Since any valid traversal from $(1,1)$ to $(n,m)$ passing through~$B$ intersects~$B_r$, the output $B_t^R$ is not needed anymore, and we are done with $B$. A symmetric rule applies to boxes which touch the right boundary of the free-space diagram.
%See Figure \ref{fig:boundary_pruning} for a visualization.

% \begin{figure}
%     % \includegraphics[width=\textwidth]{figures/boundary_pruning.pdf}
%     \caption{}
%     \label{fig:shrink_box}
% \end{figure}

\subsection{Implementation Details of Simple Boundaries} \label{subsec:simple_boundaries_impl}
It remains to describe how we test whether a boundary is simple, and how we determine the free interval of a simple boundary.
One important ingredient for the fast detection of simple boundaries are two simple heuristic checks that check whether two polygonal curves are close or far, respectively. The former check was already used in \cite{sigspatial1}. We first explain these heuristic checks, and then explain how to use them for the detection of simple boundaries.

\paragraph{Heuristic check whether two curves are close.}
Given two subcurves $\pi' \coloneqq \pi_{i \dots i'}$ and $\sigma' \coloneqq \sigma_{j \dots j'}$, this filter heuristically tests whether $d_F(\pi', \sigma') \leq \delta$. Let $i_c \coloneqq \lfloor\frac{i+i'}{2}\rfloor$ and $j_c \coloneqq \lfloor\frac{j+j'}{2}\rfloor$ be the indices of the midpoints of $\pi'$ and $\sigma'$ (with respect to hops). Then $d_F(\pi', \sigma') \leq \dist$ holds if
\[
	%\text{longer half of $\pi$} + \text{distance of midpoints} + \text{longer half of $\sigma$} = 
	\max\{\norm{\pi_{i \dots i_c}}, \norm{\pi_{i_c \dots i'}}\} + \norm{\pi_{i_c} - \sigma_{j_c}} + \max\{\norm{\sigma_{j \dots j_c}}, \norm{\sigma_{j_c \dots j'}}\} \leq \dist.
\]
The triangle equality ensures that this is an upper bound on all distances between two points on the curves. For a visualization, see Figure \ref{fig:heuristic_check_close}. Observe that all curve lengths that need to be computed in the above equation can be determined quickly due to our preprocessing, see Section \ref{sec:preprocessing}. We call this procedure $\HeurClose(\pi', \sigma', \dist)$. 

\begin{figure}
	\begin{subfigure}[b]{0.38\textwidth}
		\includegraphics[width=\textwidth]{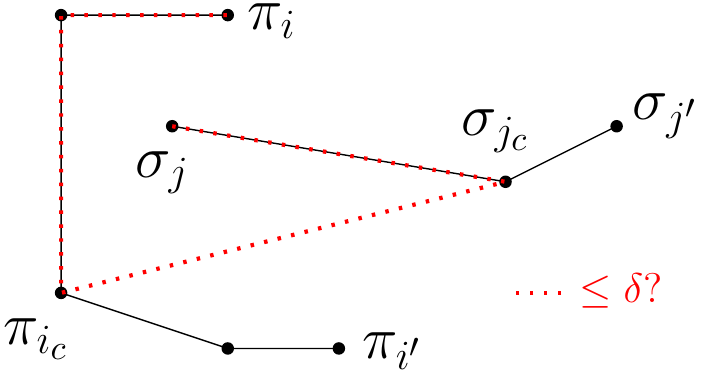}
		\subcaption{\HeurClose{}}
		\label{fig:heuristic_check_close}
	\end{subfigure}
	\hfill
	\begin{subfigure}[b]{0.55\textwidth}
		\includegraphics[width=\textwidth]{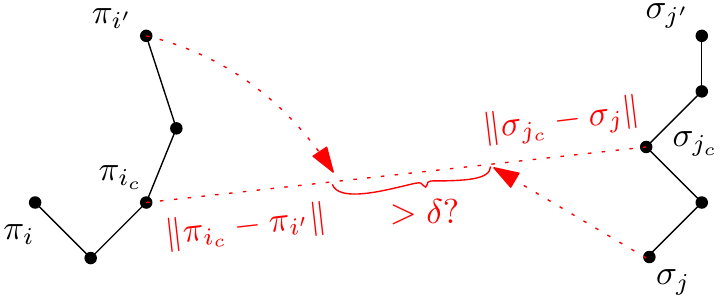}
		\subcaption{\HeurFar{}}
		\label{fig:heuristic_check_far}
	\end{subfigure}
	\caption{Visualizations of heuristic checks \HeurClose{} and \HeurFar{}.}
\end{figure}

\paragraph{Heuristic check whether two curves are far.} Symmetrically, we can test whether all pairs of points on $\pi'$ and $\sigma'$ are far by testing
\[
	\norm{\pi_{i_c} - \sigma_{j_c}} - \max\{\norm{\pi_{i \dots i_c}}, \norm{\pi_{i_c \dots i'}}\} - \max\{\norm{\sigma_{j \dots j_c}}, \norm{\sigma_{j_c \dots j'}}\} > \dist.
\]
We call this procedure $\HeurFar(\pi', \sigma', \dist)$.

\paragraph{Computation of Simple Boundaries.} Recall that an interval is defined as $I = \{p\} \times [q,q']$ (intervals of the form $[q,q'] \times \{p\}$ are handled symmetrically).
The naive way to decide whether interval $I$ is simple would be to go over all the segments of $\sigma_{q \dots q'}$ and compute the intersection with the circle of radius \dist{} around $\pi_p$. However, this is too expensive because
\begin{enumerate*}[label=(\roman*)]
	\item computing the intersection of a disc and a segment involves taking a square root, which is an expensive operation with a large constant running time, and\label{list:simple_interval_comp1}
	\item iterating over all segments of $\sigma_{q \dots q'}$ incurs a linear factor in $n$ for large boxes, while we aim at a logarithmic dependence on~$n$ for simple boundary detection.\label{list:simple_interval_comp2}
\end{enumerate*}
%Note that those are two different types of issues: while \ref{list:simple_interval_comp1} implies a large constant, \ref{list:simple_interval_comp2} is an issue of asymptotic running time; it is linear, but we aim for logarithmic.

We avoid these issues by resolving long subcurves $\sigma_{j..j+s}$ using our heuristic checks (\HeurClose{}, \HeurFar{}). Here, $s$ is an adaptive step size that grows whenever the heuristic checks were applicable, and shrinks otherwise. 
See Algorithm \ref{alg:get_simple_interval} for pseudocode of our simple boundary detection. It is straightforward to extend this algorithm to not only detect whether a boundary is simple, but also compute the free interval of a simple boundary; we call the resulting procedure  \textsc{SimpleBoundary}.

\begin{algorithm}[t]
\begin{algorithmic}[1]
\Procedure{isSimpleBoundary}{$\pi_p, \sigma_{q \dots q'}$}
\If {$\HeurFar(\pi_p, \sigma_{q \dots q'}, \dist)$ or $\HeurClose(\pi_p, \sigma_{q \dots q'}, \dist)$}
\State \Return \enquote{simple}
\EndIf
\State
% \State $\text{saw\_free} \gets \text{true if } \norm{p - \pi_1} \leq \dist$, false otherwise
\State $C \gets
\begin{cases}
	\{\sigma_q\} & ,\text{if } \norm{p - \sigma_q} \leq \dist \\
	\emptyset & ,\text{otherwise}
\end{cases}$ \Comment{set of change points}
\State $s \gets 1, j \gets q$
\While {$j < q'$}
\If {$\HeurClose(\pi_p, \sigma_{j \dots j+s}, \dist)$} \label{line:heurb}
		\State $j \gets j+s$
		\State $s \gets 2s$
		\ElsIf {$\HeurFar(\pi_p, \sigma_{j \dots j+s}, \dist)$} \label{line:heurc}
		\State $j \gets j+s$
		\State $s \gets 2s$
	\ElsIf {$s > 1$}
		\State $s \gets s/2$
	\Else
		\State $P \gets \{ j' \in (j,j+1] \mid \norm{\pi_p - \sigma_{j'}} = \dist\}$ %\Comment{by intersecting $\pi_p$ with $\sigma_{j \dots j+1}$}
		\State $C \gets C \cup P$
		\State $j \gets j+1$
		\If {$\abs{C} > 2$}
			\State \Return \enquote{not simple}
		\EndIf
	\EndIf
\EndWhile
\State
\State \Return \enquote{simple}
\EndProcedure
\end{algorithmic}
\caption{Checks if the boundary in the free-space diagram corresponding to $\{p\} \times [q,q']$ is simple.}
\label{alg:get_simple_interval}
\end{algorithm}

\subsection{Effects of Combined Pruning Rules}
All the pruning rules presented above can in practice lead to a reduction of the number of boxes that are necessary to decide the Fréchet distance of two curves. We exemplify this on two real-world curves; see Figure~\ref{fig:freespace_diagram_light} on page~\pageref{fig:freespace_diagram_light} for the curves and their corresponding free-space diagram. We explain in the following where the single rules come into play. For \emph{Box 1} we apply Rule \rom{3}b twice -- for the top and right output. The top boundary of \emph{Box 2} is empty and thus we computed the outputs according to Rule \rom{3}a. Note that the right boundary of this box is on the right boundary of the free-space diagram and thus we do not have to compute it according to Rule \rom{4}. For \emph{Box 3} we again use Rule \rom{3}b for the top, but we use Rule \rom{3}c for the right boundary -- the blue dotted line indicates that the reachability information is propagated through the box. For \emph{Box 4} we first use Rule \rom{2} to move the bottom boundary significantly up, until the end of the left empty part; we can do this because the bottom boundary is empty and the left boundary is simple, starting with an empty part. After two splits of the remaining box, we see that the two outputs of the leftmost box are empty as the top and right boundaries are non-free, using Rule \rom{3}a. For the remaining two boxes we use Rule \rom{1} as their inputs are empty.

\begin{figure}
	\includegraphics[width=\textwidth]{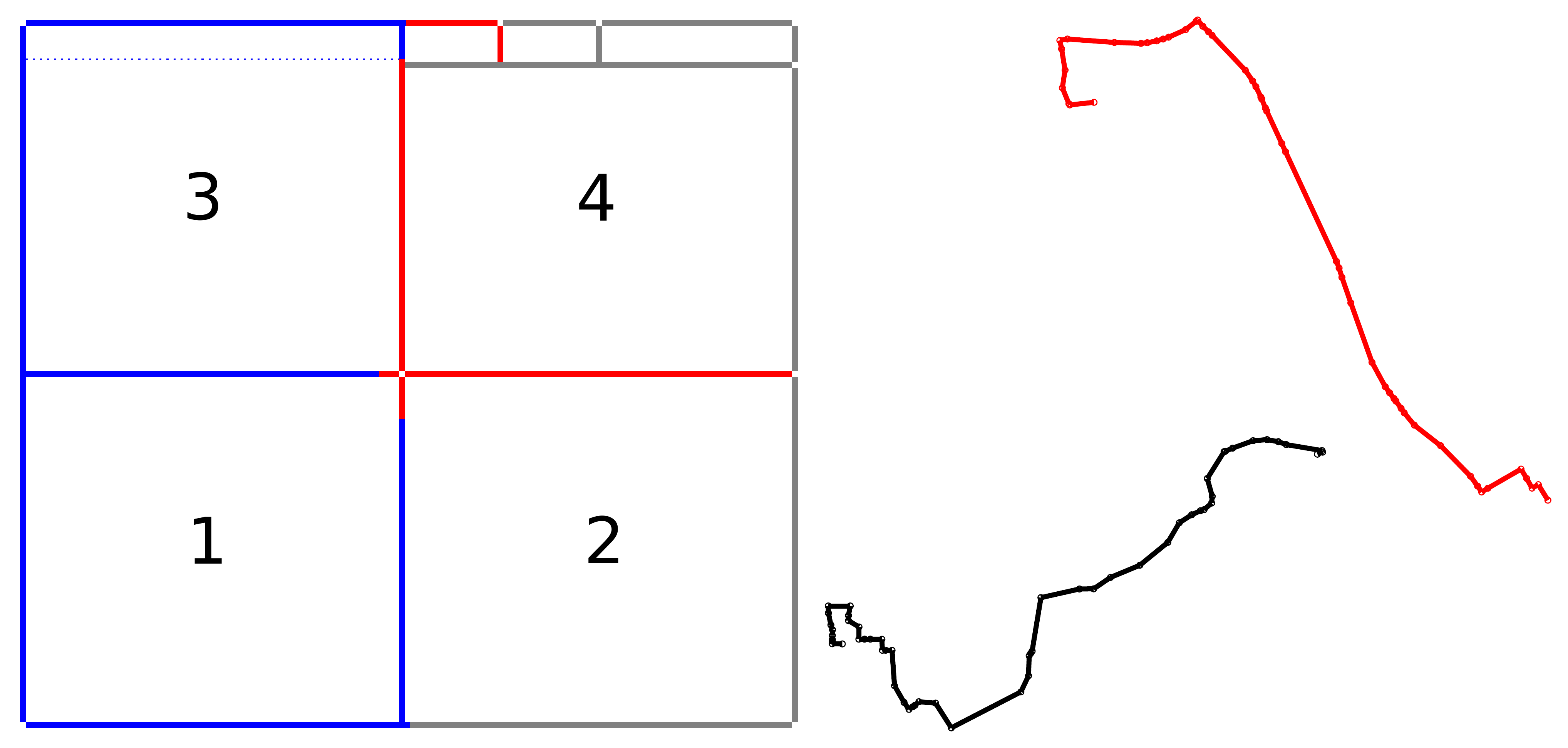}
	\caption{A free-space diagram as produced by our final implementation (left) with the corresponding curves (right). The curves are taken from the SIGSPATIAL dataset. We number the boxes in the third level of the recursion from 1 to 4.}
	\label{fig:freespace_diagram_light}
\end{figure}

This example illustrates how propagating through a box (in \emph{Box 3}) and subsequently moving a boundary (in \emph{Box 4}) leads to pruning large parts. Additionally, we can see how using simple boundaries leads to early decisions and thus avoids many recursive steps. In total, we can see how all the explained pruning rules together lead to a free-space diagram with only twelve boxes, \ie, twelve recursive calls, for curves with more than 50 vertices and more than 1500 reachable cells.
% \footnote{The curves shown in Figure \ref{fig:freespace_diagram_light} are subsampled variants of two curves from the \emph{Sigspatial} data set. Even without the subsampling, our algorithm performs only twelve recursive calls.}
Figure \ref{fig:pruning_rules_example} shows what effects the pruning rules have by introducing them one by one in an example.

\begin{figure}
	% \vspace{-2cm}
	\centering
	\includegraphics[width=.8\textwidth]{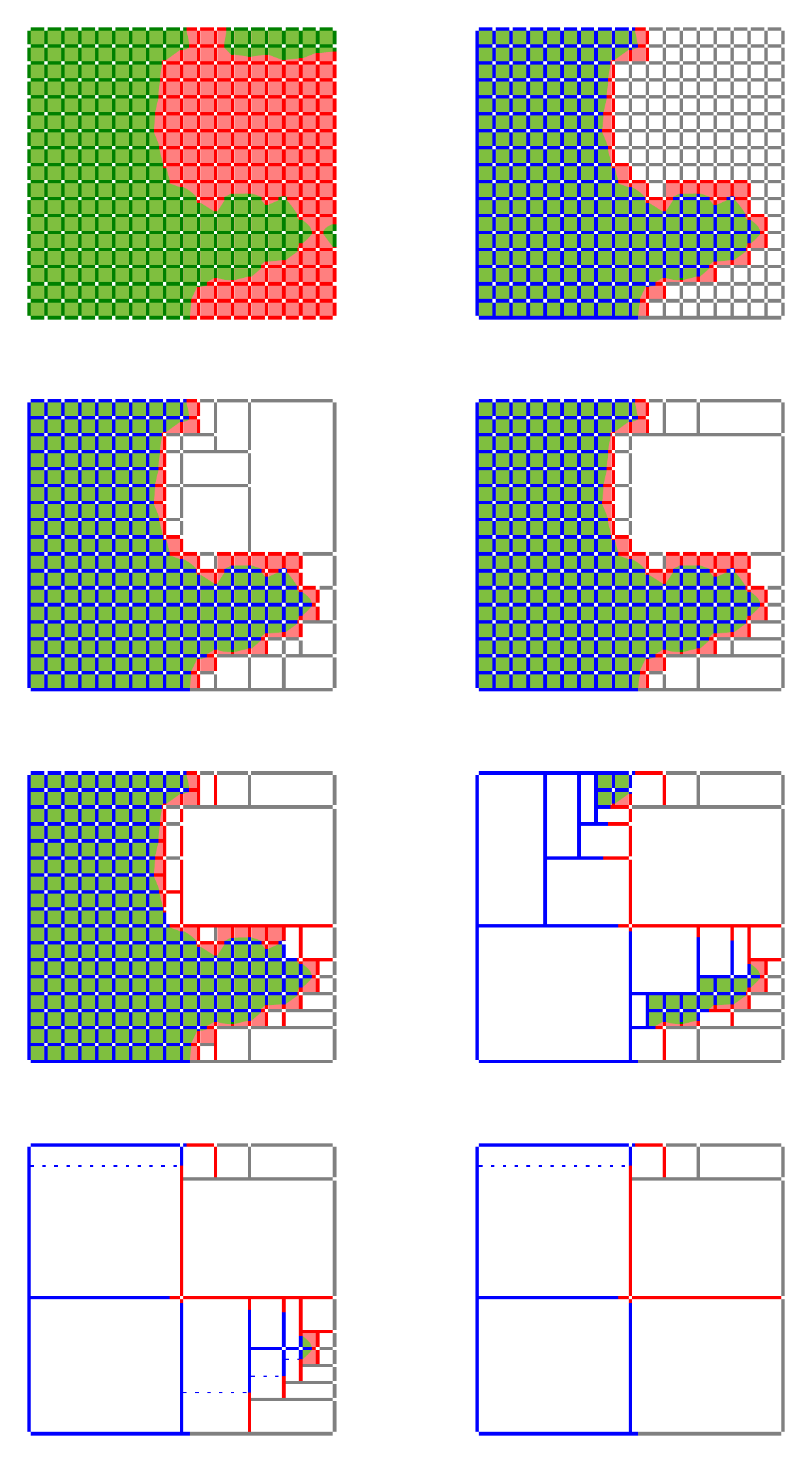}
	\caption{A decider example introducing the pruning rules one by one. They are introduced from top to bottom and left to right. The images in this order depict: the free-space diagram, the reachable space, after introducing Rule \rom{1}, Rule \rom{2}, Rule \rom{3}a, Rule \rom{3}b, Rule \rom{3}c, Rule \rom{4}, and finally the free-space diagram with all pruning rules enabled. The curves of this example are shown in Figure \ref{fig:freespace_diagram_light}.}
	\label{fig:pruning_rules_example}
\end{figure}

%% file: trunk/decider.tex
\section{Decider with Filters} \label{sec:decider}
Now that we introduced the complete decider, we are ready to present the decider. We first give a high-level overview.
\subsection{Decider}
The decider can be divided into two parts:
\begin{enumerate}
	\item Filters (see this section)
	\item Complete decider via free-space exploration (see Section~\ref{sec:complete_decider})
\end{enumerate}
As outlined in Section~\ref{sec:intro}, we first try to determine the correct output by using fast but incomplete filtering mechanisms and only resort to the slower complete decider presented in the last section if none of the heuristic deciders (\emph{filters}) gave a result.
The high-level pseudocode of the decider is shown in Algorithm \ref{alg:overview}. 

\begin{algorithm}[t]
\begin{algorithmic}[1]
\Procedure{Decider}{$\pi, \sigma, \delta$}
\If {start points $\pi_1, \sigma_1$ or end points $\pi_n, \sigma_m$ are far}
\Return \enquote{far}
\EndIf
\ForAll {$f \in \text{Filters}$}
\State {$\text{verdict} = f(\pi, \sigma, \delta)$}
\If {$\text{verdict} \in \{\text{\enquote{close}}, \text{\enquote{far}\}}$}
\State \Return \text{verdict}
\EndIf
\EndFor
\State \Return $\Call{CompleteDecider}{\pi,\sigma, \delta}$
\EndProcedure
\end{algorithmic}
\caption{High-level code of the Fréchet decider.}
\label{alg:overview}
\end{algorithm}

The speed-ups introduced by our complete decider were already explained in Section~\ref{sec:complete_decider}. A second source for our speed-ups lies in the usage of a good set of filters. Interestingly, since our optimized complete decider via free-space exploration already solves many simple instances very efficiently, our filters have to be extremely fast to be useful -- otherwise, the additional effort for an incomplete filter does not pay off. In particular, we cannot afford expensive preprocessing and ideally, we would like to achieve sublinear running times for our filters. To this end, we only use filters that can traverse large parts of the curves quickly.
% E.g., on one of our benchmarks with an average number of ??? vertices per curve, our greedy filter finds feasible traversals for ??\% of the instances, while traversing each curve in only ?? steps on average. \andre{Numbers missing. Insert after new experiments conducted.}
We achieve sublinear-type behavior by making previously used filters work with an adaptive step size (exploiting fast heuristic checks), and designing a new adaptive negative filter.

In the remainder of this section, we describe all the filters that we use to heuristically decide whether two curves are close or far. There are two types of filters: \emph{positive} filters check whether a curve is close to the query curve and return either \enquote{close} or \enquote{unknown}; \emph{negative} filters check if a curve is far from the query curve and return either \enquote{far} or \enquote{unknown}.

\subsection{Bounding Box Check}
This is a positive filter already described in \cite{sigspatial3}, which heuristically checks whether all pairs of points on $\pi, \sigma$ are in distance at most $\delta$. Recall that we compute the bounding box of each curve when we read it. We can thus check in constant time whether the furthest points on the bounding boxes of $\pi, \sigma$ are in distance at most $\delta$. If this is the case, then also all points of $\pi, \sigma$ have to be close to each other and thus the free-space diagram is completely free and a valid traversal trivially exists.

\begin{figure}
	\includegraphics[width=\textwidth]{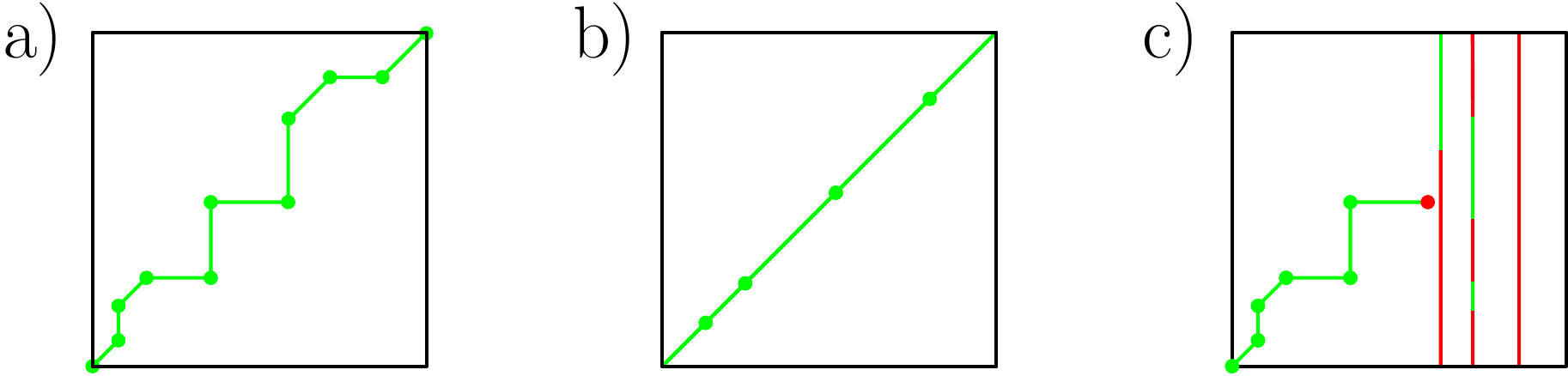}
	\caption{Sketches of the (a) greedy filter, (b) adaptive equal-time filter, and (c) negative filter. These sketches should be read as follows: the first dimension is the index on the first curve, while the second dimension is the index on the second curve. The green color means that the corresponding points are in distance at most $\delta$. Otherwise they are colored red. This visualization is similar to the free-space diagram.}
	\label{fig:filter_sketches}
\end{figure}

\subsection{Greedy}
This is a positive filter. To assert that two curves $\pi$ and $\sigma$ are close, it suffices to find a traversal $(f,g)$ satisfying $\max_{t \in [0,1]} \norm{\pi_{f(t)} - \sigma_{g(t)}} \leq \dist$. We try to construct such a traversal staying within distance $\dist$ by making greedy steps that minimize the current distance.
This may yield a valid traversal: if after at most $n+m$ steps we reach both endpoints and during the traversal the distance was always at most \dist{}, we return \enquote{near}. We can also get stuck: if a step on each of the curves would lead to a distance greater than $\dist$, we return \enquote{unknown}.
A similar filter was already used in \cite{sigspatial1}, however, here we present a variant with adaptive step size. This means that instead of just advancing to the next node in the traversal, we try to make larger steps, leveraging the heuristic checks presented in Section \ref{subsec:simple_boundaries_impl}. We adapt the step size depending on the success of the last step.
For pseudocode of the greedy filter see Algorithm \ref{alg:greedy_filter}, and for a visualization see Figure \ref{fig:filter_sketches}a.

\begin{algorithm}[t]
\begin{algorithmic}[1]
\Procedure{GreedyFilter}{$\pi, \sigma, \dist$}
\State $i, j, s \gets 1$
\While {$i < n$ or $j < m$}
	\State $S \gets
	\begin{cases}
		\{(i+1,j),(i,j+1),(i+1,j+1)\},& \text{if $s = 1$} \\
		\{(i+s,j),(i,j+s)\},& \text{if $s > 1$} \\
	\end{cases}$
	\Comment{possible steps}
	\State $P \gets \{(i',j') \in S \mid i' \leq n \ \&\  j' \leq m \ \&\  \HeurClose(\pi_{i \dots i'}, \sigma_{j \dots j'}, \dist)\}$
	\If {$P = \emptyset$}
		\If {$s=1$}
			\State \Return \enquote{unknown}
		\Else
			\State $s \gets s/2$
		\EndIf
	\Else
		\State $(i,j) \gets \argmin_{(i',j') \in P} \norm{\pi_{i'} - \sigma_{j'}}$
			\State $s \gets 2s$
	\EndIf
\EndWhile
\State \Return \enquote{close}
\EndProcedure
\end{algorithmic}
\caption{Greedy filter with adaptive step size.}
\label{alg:greedy_filter}
\end{algorithm}

\subsection{Adaptive Equal-Time}
We also consider a variation of Greedy Filter, which we call Adaptive Equal-Time Filter. The only difference to Algorithm~\ref{alg:greedy_filter} is that the allowed steps are now:
\[
	S \coloneqq
	\begin{cases}
		\{(i+1,j),(i,j+1),(i+1,j+1)\},& \text{if $s = 1$}, \\
		\left\{\left(i+s, j+ \left\lfloor \frac{m-j}{n-i} \cdot s \right\rfloor\right)\right\},& \text{if $s > 1$.} \\
	\end{cases}
\]
In contrast to Greedy Filter, this searches for a traversal that stays as close as possible to the diagonal.

\subsection{Negative}
A negative filter was already used in \cite{sigspatial1} and \cite{sigspatial3}. However, changing this filter to use an adaptive step size does not seem to be practical when used with our approach. Preliminary tests showed that this filter would dominate our running time. Therefore, we developed a new negative filter which is more suited to be used with an adaptive step size and thus can be used with our approach.

% Intuitively, we search for a point on $\pi$ that is far from each point on $\sigma$.
Let $(\pi_i, \sigma_j)$ be the points at which Greedy Filter got stuck. We check whether some point $\pi_{i+2^k}, k \in \mathbb{N}$, is far from all points of $\sigma$ using \HeurFar. If so, we conclude that $d_F(\pi,\sigma) > \dist$. We do the same with the roles of $\pi$ and $\sigma$ exchanged.
% To search for a far point on $\pi$ we use an exponentially increasing step size starting from $\pi_1$ to choose the point $\pi_i$ and then use an adaptive step size to iterate over $\sigma$ to check if the whole curve is far from $\pi_i$, always performing the heuristic check presented in Section \ref{subsec:simple_boundaries_impl}. Instead of starting from $\pi_1$, in the improved version we start from the point at which greedy got stuck. We already know that we can reach this point with a valid traversal and thus checking the subcurve before that point can never lead to \enquote{far} being returned.
% As this filter is not symmetric, we can use it a second time with $\pi$ and $\sigma$ being swapped.
See Algorithm \ref{alg:negative_filter} for the pseudocode of this filter; for a visualization see Figure \ref{fig:filter_sketches}c.

\begin{algorithm}[t]
\begin{algorithmic}[1]
\Procedure{NegativeFilter}{$\pi, \sigma, \dist$}
\State $(i,j) \gets \text{ last indices of close points in greedy filter}$
\State $s \gets 1$ \label{line:start_neg}
\While {$i + s \leq n$}
	\If {$\textsc{SimpleBoundary}(\pi_{i+s}, \sigma, \dist)$ is non-free}
		\State \Return \enquote{far}
	\EndIf
	\State $s \gets 2s$
\EndWhile \label{line:end_neg}
\State
\State Repeat lines \ref{line:start_neg} to \ref{line:end_neg} with the roles of $\pi$ and $\sigma$ swapped
\State
\State \Return \enquote{unknown}
\EndProcedure
\end{algorithmic}
\caption{Negative filter, where in the two if statements we do a search with adaptive step size on $\sigma$ and $\pi$, respectively.}
\label{alg:negative_filter}
\end{algorithm}

%% file: trunk/query.tex
\section{Query Data Structure} \label{sec:query}

In this section we give the details of extending the fast decider to compute the Fréchet distance in the query setting. Recall that in this setting we are given a \emph{curve dataset} $\mathcal{D}$ that we want to preprocess for the following queries: Given a polygonal curve $\pi$ (the \emph{query curve}) and a threshold distance $\delta$, report all $\sigma \in \mathcal{D}$ that are $\delta$-close to $\pi$. To be able to compare our new approach to existing work (especially the submissions of the GIS Cup) we present a query data structure here, which is influenced by the one presented in \cite{sigspatial1}.

The most important component that we need additionally to the decider to obtain an efficient query data structure is a mechanism to quickly determine a set of candidate curves on which we can then run the decider presented above. The candidate selection is done using a kd-tree on 8-dimensional points, similar to the octree used in~\cite{sigspatial1}, see \ref{sec:kdtree} for more details. The high-level structure of the algorithm for answering queries is shown in Algorithm \ref{alg:find_close_curves}.

\begin{algorithm}[t]
\begin{algorithmic}[1]
\Procedure{FindCloseCurves}{$\pi, \delta$}
% \State Preprocess $\pi$ \label{line:preprocess_pi}
\State $C \gets \Call{kdtree.query}{\pi, \delta}$
\State $R \gets \emptyset$ 
\ForAll {$\sigma \in C$}
\If {$\Call{FrechetDistanceDecider}{\pi, \sigma, \delta} = \text{\enquote{close}}$}
\State $R \gets R \cup \{\sigma\}$
\EndIf
\EndFor
\State \Return $R$
\EndProcedure
\end{algorithmic}
\caption{The function for answering a range query.}
\label{alg:find_close_curves}
\end{algorithm}

% There are two points at which preprocessing is necessary. First, when we read the data set $\mathcal{D}$, we compute all the prefix distances for $\sigma \in \mathcal{D}$ as described in \ref{sec:preprocessing} and subsequently build the kd-tree. Second, we also need to calculate the prefix distances and bounding box for the query curve $\pi$ when we run a query. This is done in line \ref{line:preprocess_pi} of Algorithm \ref{alg:find_close_curves}.

\subsection{Kd-Tree} \label{sec:kdtree}
Fetching an initial set of candidate curves via a space-partitioning data structure was already used in \cite{sigspatial1, sigspatial2, sigspatial3}. We use a kd-tree which contains 8-dimensional points, each corresponding to one of the curves in the data set. Four dimensions are used for the start and end point of the curve and the remaining four dimensions are used for the maximum/minimum coordinates in x/y direction. We can then query this kd-tree with the threshold distance \dist{} and obtain a set of candidate curves. See \cite{sigspatial1} for why proximity in the kd-tree is a necessary condition for every close curve.

%% file: trunk/implementation_details.tex
\section{Implementation Details} \label{sec:implementation_details}

\paragraph{Square Root.} Computing which parts are close and which are far between a point and a segment involves intersecting a circle and a line segment, which in turn requires computing a square root.
As square roots are computationally quite expensive, we avoid them by:
\begin{itemize}
	\item filtering out simple comparisons by heuristic checks not involving square roots
	\item testing $x < a^2$ instead of $\sqrt{x} < a$ (and analogous for other comparisons)
\end{itemize}
While these changes seem trivial, they have a significant effect on the running time due to the large amount of distance computations in the implementation.

\paragraph{Recursion.} Note that the complete decider (Algorithm \ref{alg:recursive_decider}) is currently formulated as a recursive algorithm. Indeed, our implementation is also recursive, which is feasible due to the logarithmic depth of the recursion. An iterative variant that we implemented turned out to be equally fast but more complicated, thus we settled for the recursive variant.

% \paragraph{Floating Point Issues.} \ANDRE{write something}

%% file: trunk/experiments.tex
\section{Experiments} \label{sec:experiments}
% In the experiments, we aim to substantiate the following two claims. First, we want to compare our decider algorithm to the previously best known implementation to verify that our main contribution actually is a significant improvement over the state of the art. To the best of our knowledge the winning submission of the \giscup is the currently fastest implementation for a decider of the continuous Fréchet distance. Among the top three submissions it is the only one which actually improved the classical dynamic program using a free-space diagram by \cite{alt_95}. Second, we want to compare to the results of the \giscup to also show that we significantly improve the state of the art in the query setting -- we consider this necessary as the \giscup probably attracted the best current implementations in the research community.
In the experiments, we aim to substantiate the following two claims. First, we want to verify that our main contribution, the decider, actually is a significant improvement over the state of the art. To this end, we compare our implementation with the -- to our knowledge -- currently fastest Fréchet distance decider, namely \cite{sigspatial1}. Second, we want to verify that our improvements in the decider setting also carry over to the query setting, also significantly improving the state of the art. To show this, we compare to the top three submissions of the \giscup.

% We want to highlight two additional reasons for comparing to \cite{sigspatial1} instead of another decider of the \giscup. We specifically want to point out that both other implementations naturally did not optimize for the decider setting but for the query setting, as this was the setting of the \giscup.Thus, the implementation of \cite{sigspatial2} requires significant preprocessing for simplifying the curves in the data set ($\approx20s$ on the data set of the \giscup) which makes it unsuited to \emph{directly} be used as a decider. Additionally, their improvements to the free-space algorithm rely on the simplifications computed in the preprocessing. The implementation of \cite{sigspatial3} does not require expensive preprocessing, however, they do not notably improve the algorithm of \cite{alt_95} and thus we expect a disadvantage of their implementation on difficult instances -- which we see in our query setting experiments below.

We use three different data sets: the \giscup set (\sigspatial) \cite{sigspatial_dataset}, the handwritten characters (\characters) \cite{characters_dataset}, and the GeoLife data set (\geolife) \cite{geolife_dataset}. For all experiments, we used a laptop with an Intel i5-6440HQ processor with 4 cores and 16GB of RAM.

\paragraph{Hypotheses.} In what follows, we verify the following hypotheses:
\begin{enumerate}
	\item Our implementation is significantly faster than the fastest previously known implementation in the query and in the decider setting.
	\item Our implementation is fast on a wide range of data sets.
	\item Each of the described improvements of the decider speeds up the computation significantly.
	\item The running time of the complete decider is proportional to the number of recursive calls.
\end{enumerate}
The first two we verify by running time comparisons on different data sets. The third we verify by leaving out single pruning rules and then comparing the running time with the final implementation. Finally, we verify the fourth hypothesis by correlating the running time for different decider computations against the number of recursive calls encountered during the computation.

\subsection{Data Sets Information.}
Some properties of the data sets are shown in Table \ref{tab:datasets}. \sigspatial has the most curves, while \geolife has by far the longest. \characters is interesting as it does not stem from GPS data. By this selection of data sets, we hope to cover a sufficiently diverse set of curves.

\begin{table}
\centering
\begin{tabular}{|l|lrrr|}
\hline
data set & type & \#curves & mean hops & stddev hops \\
\hline
\sigspatial & synthetic GPS-like& 20199 & 247.8 & 154.0 \\
\characters & handwritten & 2858 & 120.9 & 21.0 \\
\geolife & GPS (multi-modal) & 16966 & 1080.4 & 1844.1 \\
\hline
\end{tabular}
\caption{Information about the data sets used in the benchmarks.}
\label{tab:datasets}
\end{table}

\paragraph{Hardware.} We used standard desktop hardware for our experiments. More specifically, we used a laptop with an Intel i5-6440HQ processor with 4 cores (2.6 to 3.1 GHz) with cache sizes 256KiB, 1MiB, and 6MiB (L1, L2, L3).

\paragraph{Code.} The implementation is written in modern C++ and only has the standard library and openMP as dependencies. The target platforms are Linux and OS X, with little work expected to adapt it to other platforms. The code was optimized for speed as well as readability (as we hope to give a reference implementation).

\subsection{Decider Setting}
In this section we test the running time performance of our new decider algorithm (Algorithm \ref{alg:overview}).
We first describe our new benchmark using the three data sets, and then discuss our experimental findings, in particular how the performance and improvement over the state of the art varies with the distance and also the \enquote{neighbor rank} in the data set.

\paragraph{Benchmark.} For the decider, we want to specifically test how the decision distance $\delta$ and how the choice of the second curve $\sigma$ influences the running time of the decider. To experimentally evaluate this, we create a benchmark for each data set $\mathcal{D}$ in the following way. We select a random curve $\pi \in \mathcal{D}$ and sort the curves in the data set $\mathcal{D}$ by their distance to $\pi$ in increasing order, obtaining the sequence $\sigma_1, \dots, \sigma_n$. For all $k \in \{1, \dots, \lfloor\log n\rfloor\}$, we
\begin{itemize}
	\item select a curve $\sigma \in \{ \sigma_{2^k}, \dots, \sigma_{2^{k+1}-1} \}$ uniformly at random\footnote{Note that for $k = \lfloor\log n\rfloor$ some curves might be undefined as possibly $2^{k+1}-1 > n$. In this case we select a curve uniformly at random from $\{\sigma_{2^k}, \dots, \sigma_n \}$.},
	\item compute the exact distance $\delta^* \coloneqq d_F(\pi, \sigma)$,
	\item for each $l \in \{-10, \dots, 0\}$, add benchmark tests $(\pi, \sigma, (1 - 2^l) \cdot \delta^*)$ and $(\pi, \sigma, (1 + 2^l) \cdot \delta^*)$.
\end{itemize}
By repeating this process for 1000 uniformly random curves $\pi \in \mathcal{D}$, we create 1000 test cases for every pair of $k$ and $l$.

\paragraph{Running Times.} First we show how our implementation performs in this benchmark. In Figure \ref{fig:runtime_decider_heatmap} we depict timings for running our implementation on the benchmark for all data sets. We can see that distances larger than the exact Fréchet distance are harder than smaller distances. This effect is most likely caused by the fact that decider instances with positive result need to find a path through the free-space diagram, while negative instances might be resolved earlier as it already becomes clear close to the lower left corner of the free-space diagram that there cannot exist such a path. Also, the performance of the decider is worse for computations on $(\pi, \sigma_i, \delta)$ when $i$ is smaller. This seems natural, as curves which are closer are more likely in the data set to actually be of similar shape, and similar shapes often lead to bottlenecks in the free-space diagram (\ie, small regions where a witness path can barely pass through), which have to be resolved in much more detail and therefore lead to a higher number of recursive calls. It follows that the benchmark instances for low $k$ and~$l$ are the hardest; this is the case for all data sets. In \characters we can also see that for $k=7$ there is suddenly a rise in the running time for certain distance factors. We assume that this comes from the fact that the previous values of $k$ all correspond to the same written character and this changes for $k=7$.

\begin{figure}
	% \centering
	\includegraphics[width=\textwidth]{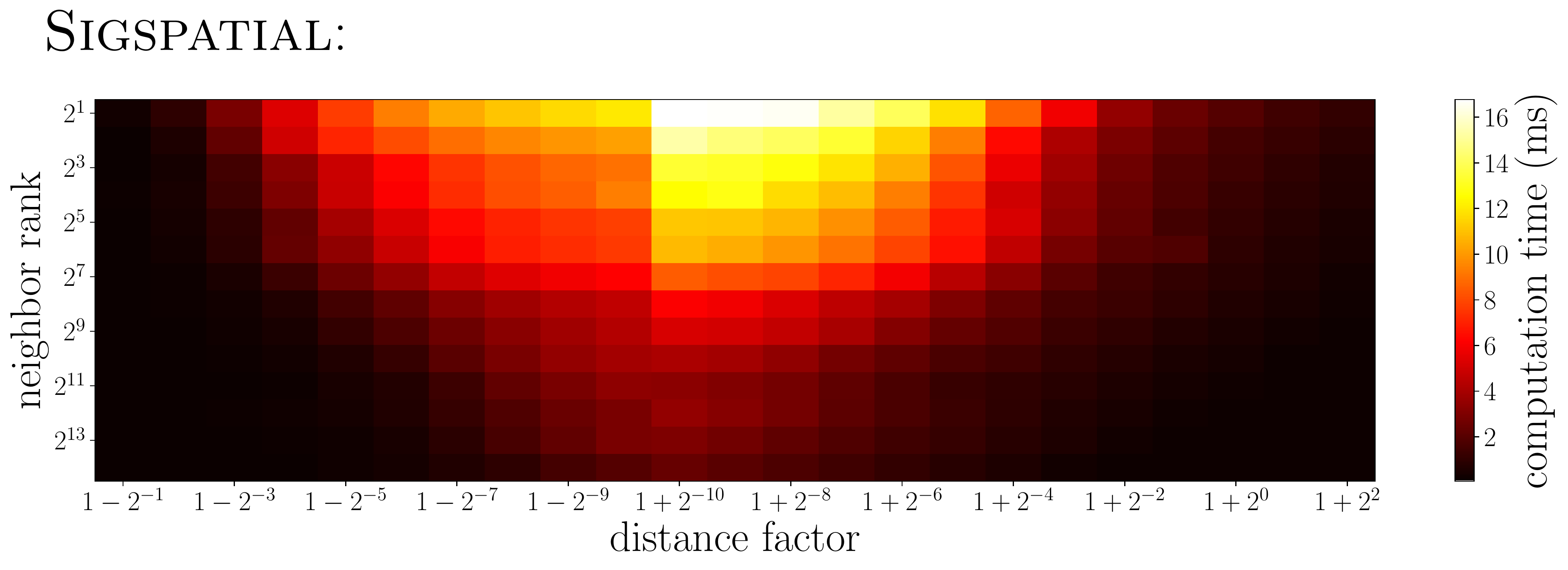}
	\includegraphics[width=\textwidth]{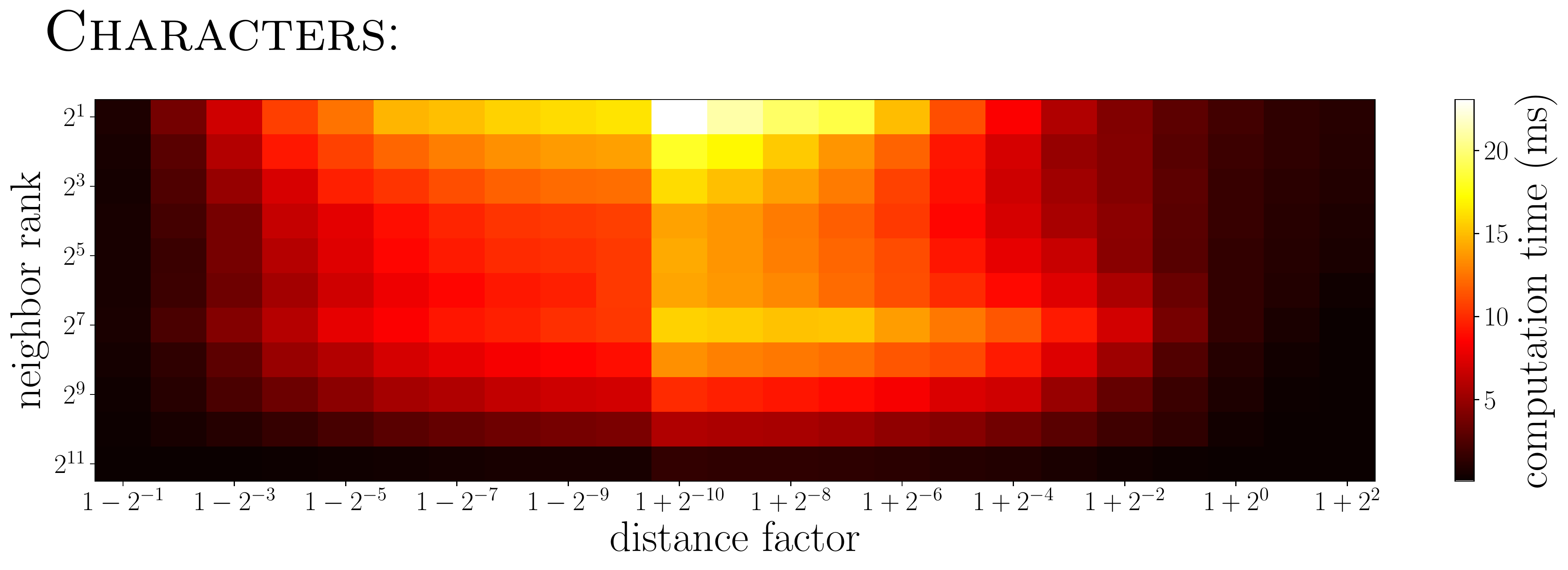}
	\includegraphics[width=\textwidth]{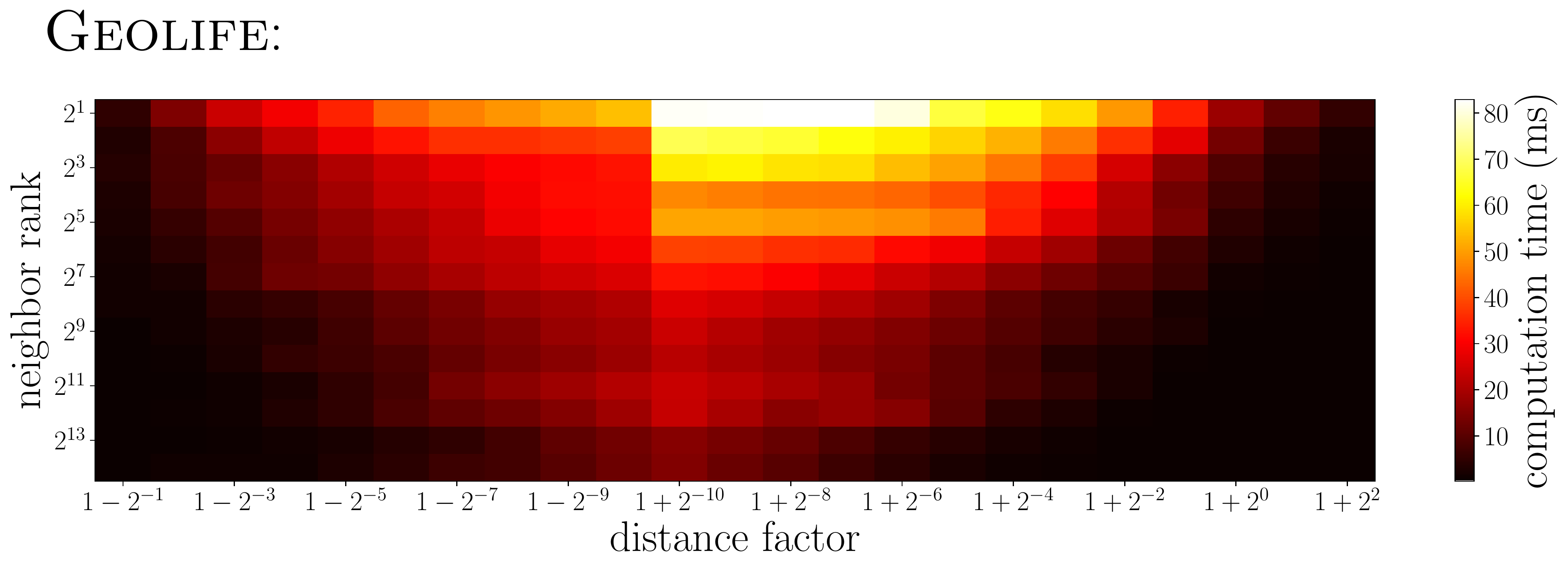}
	\caption{Running times of the decider benchmark when we run our implementation on it.}
	\label{fig:runtime_decider_heatmap}
\end{figure}

We also run the original code of the winner of the \giscup, namely \cite{sigspatial1}, on our benchmark and compare it with the running time of our implementation. See Figure \ref{fig:factor_decider_heatmap} for the speed-up factors of our implementation over the \giscup winner implementation. The speed-ups obtained depend on the data set. While for every data set a significant amount of benchmarks for different $k$ and $l$ are more than one order of magnitude faster, for \geolife even speed-ups by 2 orders of magnitude are reached. Speed-ups tend to be higher for larger distance factors. The results on \geolife suggest that for longer curves, our implementation becomes significantly faster relative to the current state of the art. Note that there also are situations where our decider shows similar performance to the one of \cite{sigspatial1}; however, those are cases where both deciders can easily recognize that the curves are far (due to, \eg, their start or end points being far). We additionally show the percentage of instances that are already decided by the filters in Figure \ref{fig:filtered_decider_heatmap}.

\begin{figure}
	% \centering
	\includegraphics[width=\textwidth]{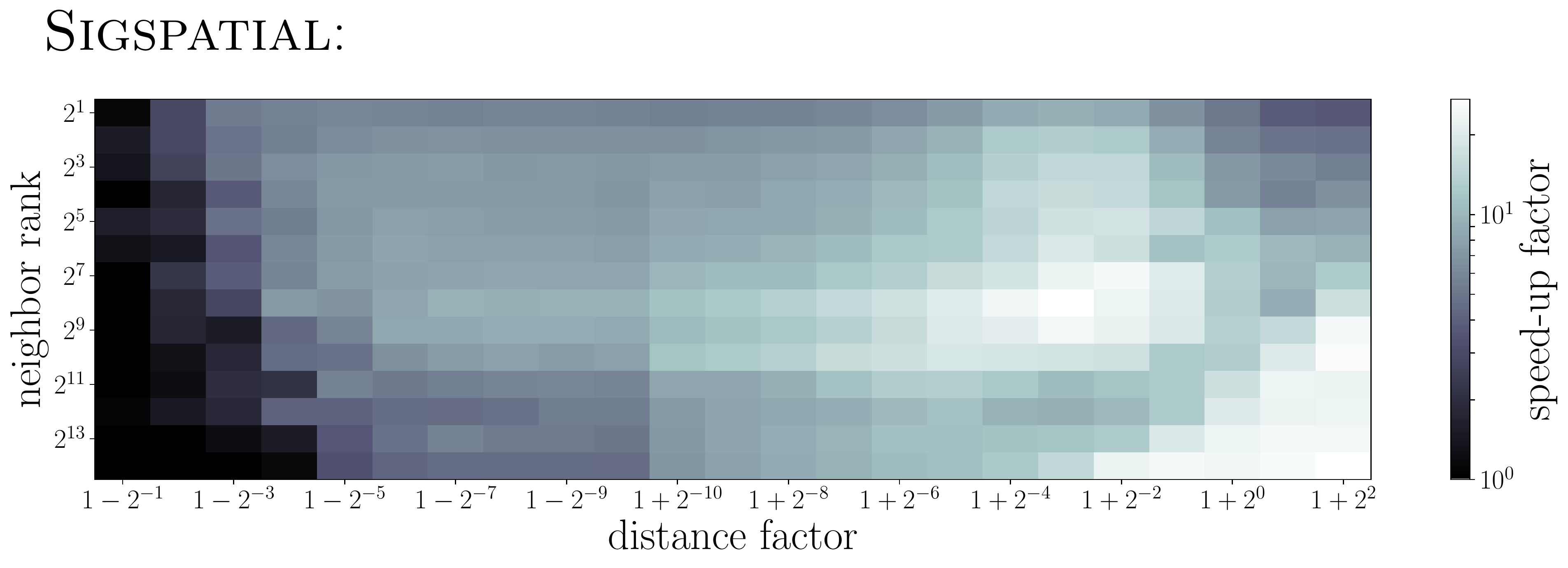}
	\includegraphics[width=\textwidth]{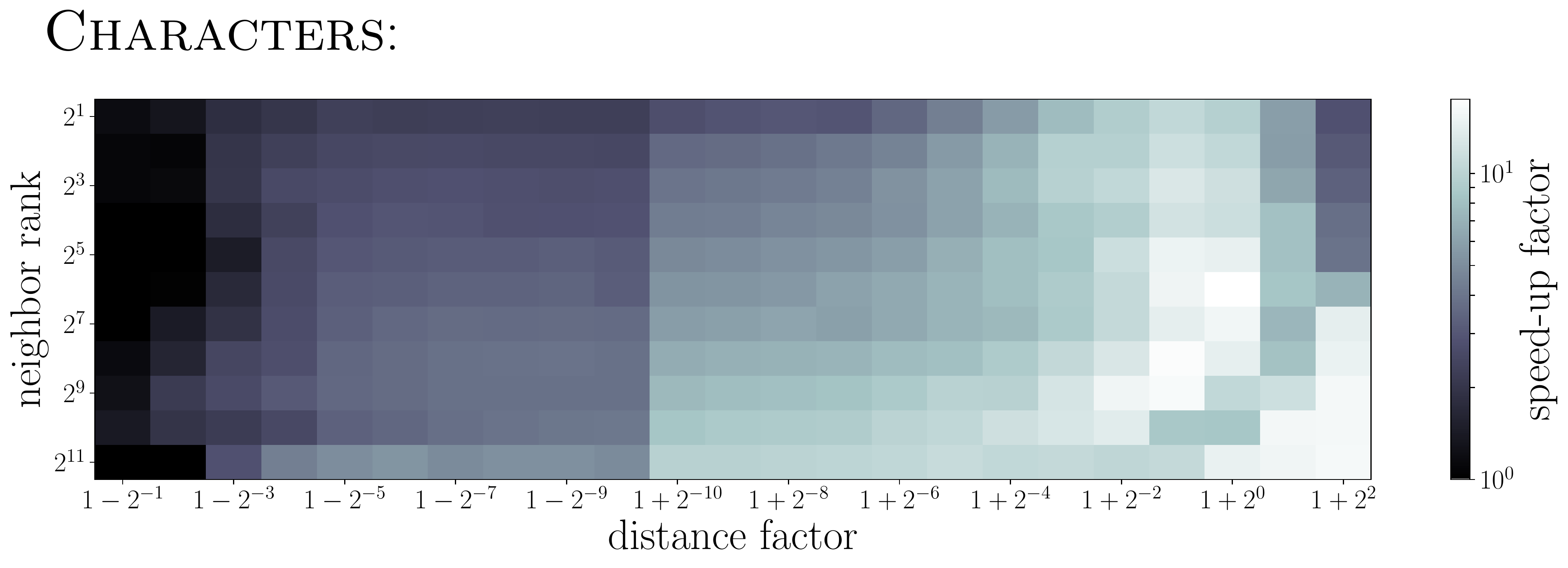}
	\includegraphics[width=\textwidth]{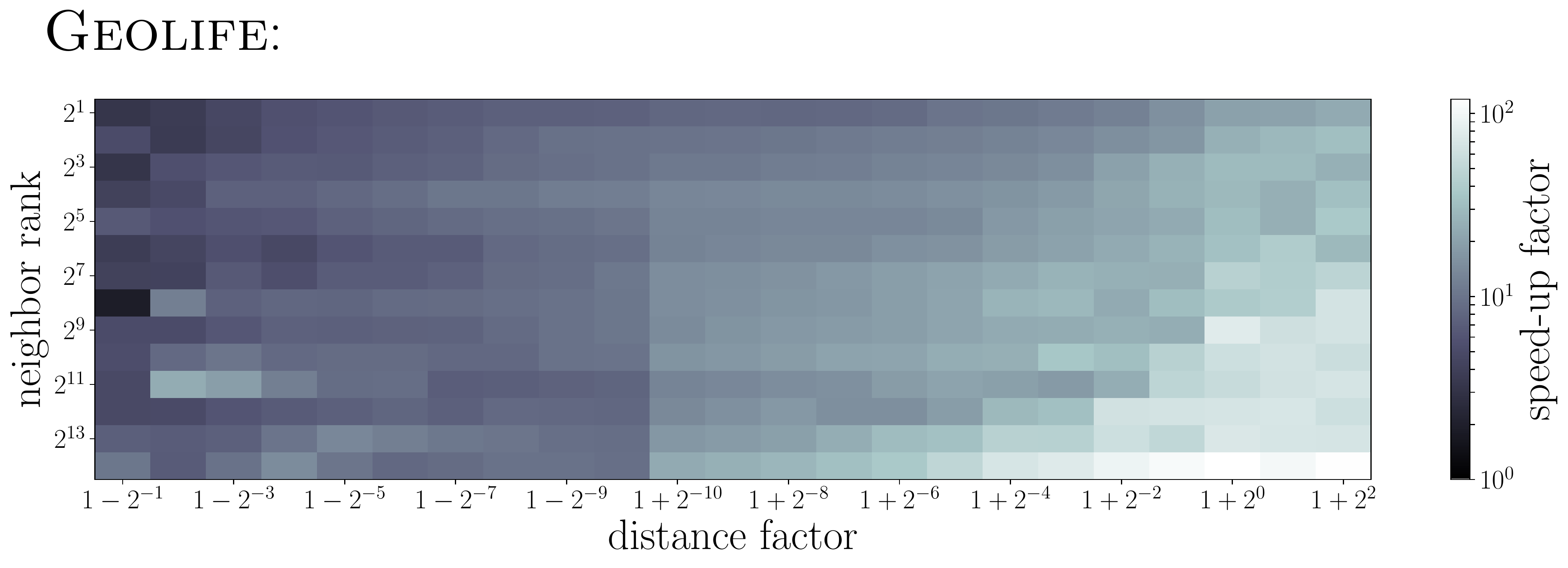}
	\caption{The speed-up factors obtained over the \giscup winner on the decider benchmark.}
	\label{fig:factor_decider_heatmap}
\end{figure}

\begin{figure}
	% \centering
	\includegraphics[width=\textwidth]{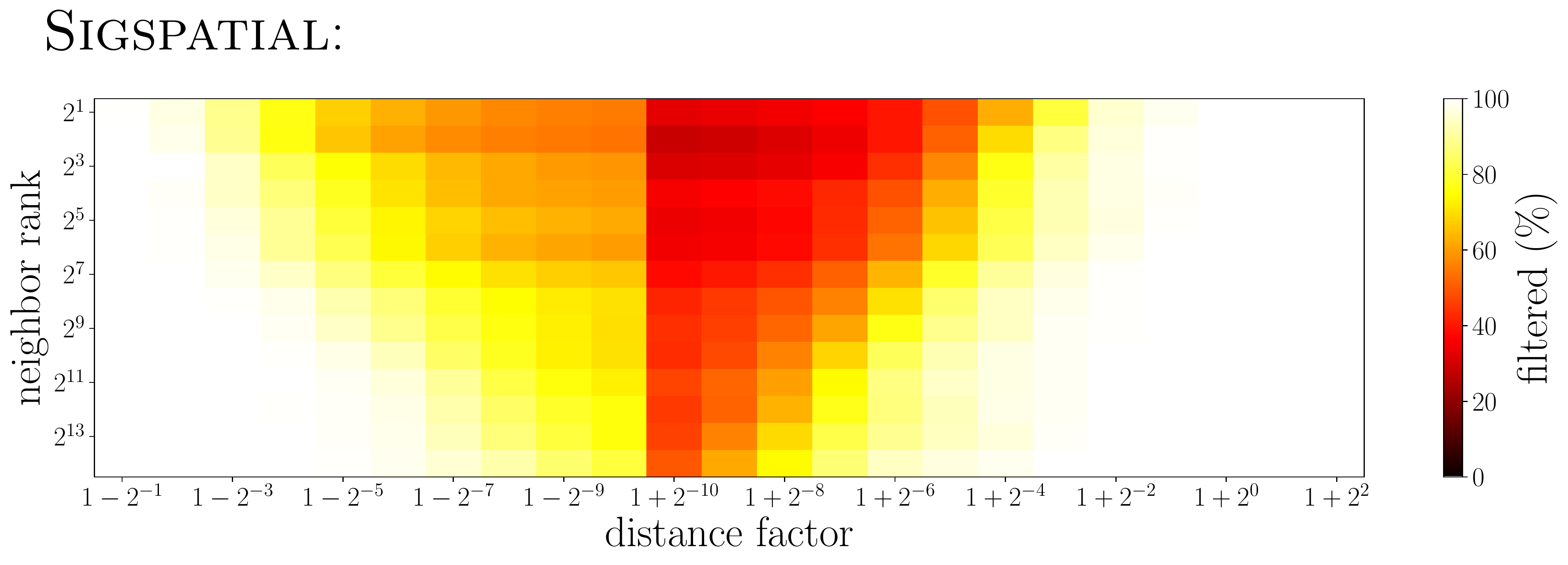}
	\includegraphics[width=\textwidth]{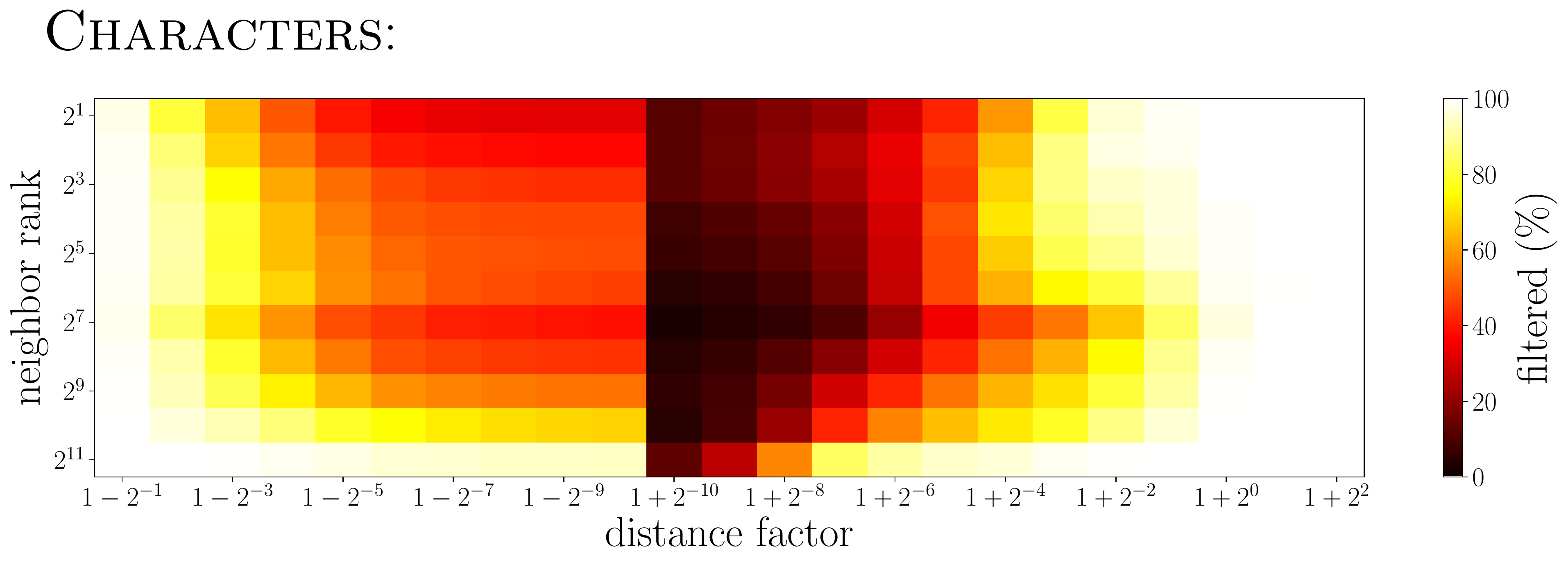}
	\includegraphics[width=\textwidth]{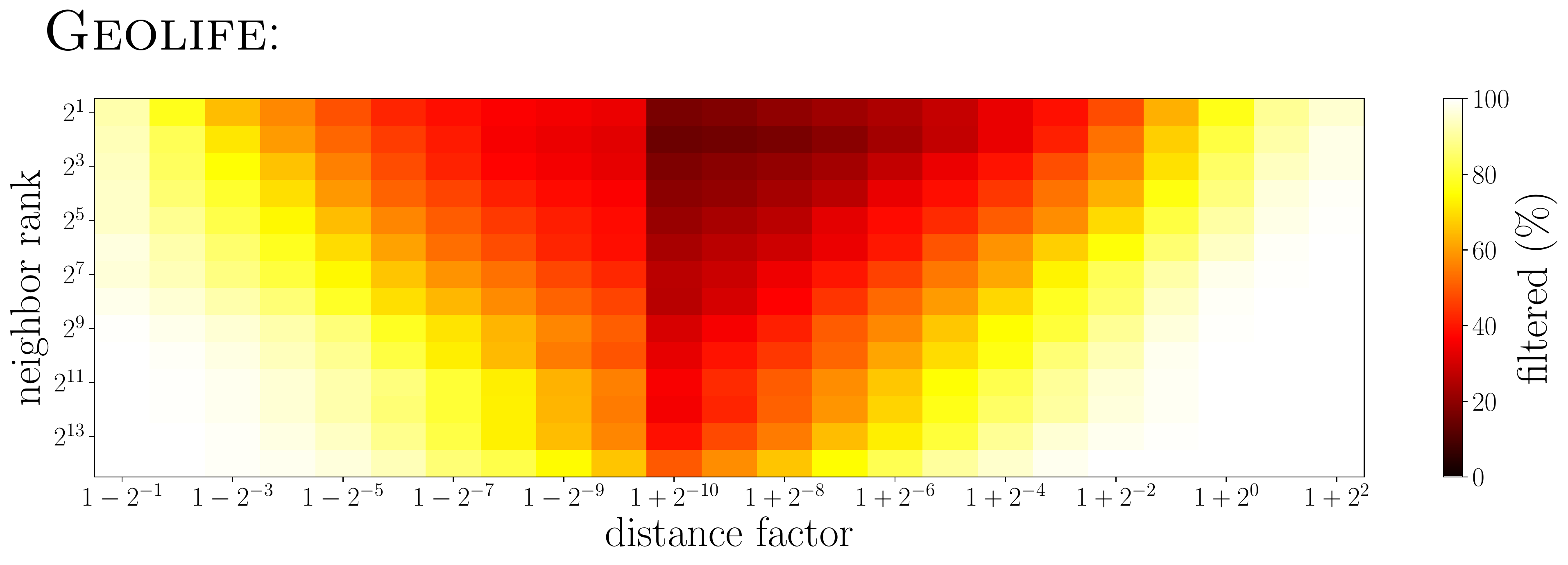}
	\caption{The percentage of queries that are decided by the filters on the decider benchmark.}
	\label{fig:filtered_decider_heatmap}
\end{figure}

\subsection{Influence of the Individual Pruning Rules}
We also verified that the improvements that we introduced indeed are all necessary. In Section \ref{sec:pruning_rules} we introduced six pruning rules. Rule \rom{1}, \ie, \enquote{Empty Inputs}, is essential. If we were to omit it, we would hardly improve over the naive free-space exploration algorithm. The remaining five rules can potentially be omitted. Thus, for each of these pruning rules, we let our implementation run on the decider benchmark with this single rule disabled; and once with all rules enabled. See Table \ref{tab:pruning_level} for the results. Clearly, all pruning rules yield significant improvements when considering the timings of the \geolife benchmark. All rules, except Rule \rom{4}, also show significant speed-ups for the other two data sets. Additionally, note that omitting Rule \rom{3}b drastically increases the running time. This effect results from Rule \rom{3}b being the main rule to prune large reachable parts, which we otherwise have to explore completely. One can clearly observe this effect in Figure \ref{fig:pruning_rules_example}.

\begin{table}
\centering
\begin{tabular}{|l|rrr|}
\hline
& \sigspatial & \characters & \geolife \\
\hline
omit none & 99.085 & 153.195 & 552.661 \\
omit Rule \rom{2} & 112.769 & 204.347 & 1382.306 \\
omit Rule \rom{3}a & 193.437 & 296.679 & 1779.810 \\
omit Rule \rom{3}b & 5317.665 & 1627.817 & 385031.421 \\
omit Rule \rom{3}c & 202.469 & 273.146 & 2049.632 \\
omit Rule \rom{4} & 110.968 & 161.142 & 696.382 \\
\hline
\end{tabular}
\caption{Times (in $ms$) for running the decider benchmarks with leaving out pruning steps. We only ran the first 100 queries for each $k$ and $l$ due to large running times when omitting the third rule.}
\label{tab:pruning_level}
\end{table}

\paragraph{Filters.} In Figure \ref{fig:filtered_decider_heatmap} we show what percentage of the queries are decided by the filters. We can see that the closer we get to the actually distance $\delta^*$ of two curves, the less likely it gets that the filters can make a decision. Furthermore, for the distances that are greater than $\delta^*$ the filters perform worse than for distances less than $\delta^*$. We additionally observe that on \characters the filters perform significantly worse than on the other two data sets. Also the running times are inversely correlated with the percentage of decisions of the filters as returning earlier in the decider naturally reduces the overall runtime.

% \paragraph{Filters.} In Table \ref{tab:adaptive_filter_steps} we show to what degree our adaptive filters can reduce the number of steps that are necessary on successful calls to the filters. \andre{Die Zahlen sind jetzt nicht so toll -- sind aber auch nur obere Schranken. Genauere Zahlen, raus nehmen oder lassen?}
% Marvin's rough suggestion: To this end, we only use filters that can traverse large parts of the curves quickly. E.g., on one of our benchmarks with an average number of 300 vertices per curve, our greedy filter finds feasible traversals for 15\% of the instances, while traversing each curve in only 4.5 steps on average.

% \begin{table}
% \centering
% \begin{tabular}{|l|rrr|}
% \hline
% & \sigspatial & \characters & \geolife \\
% \hline
% greedy filter & 12\% & 20\% & 18\% \\
% adaptive equal-time filter & 35\% & 55\% & 48\% \\
% \hline
% \end{tabular}
% \caption{Upper bounds on the percentage of steps we need to do on successful calls on the decider benchmark with these two adaptive filters as opposed to a unit step size. The percentage is computed with respect to the maximal curve length of any of the two curves, which is a naive lower bound on the number of steps a filter with unit step size has to perform.}
% \label{tab:adaptive_filter_steps}
% \end{table}

\subsection{Query Setting}
We now turn to the experiments conducted for our query data structure, which we explained in Section \ref{sec:query}.
% Our query data structure is influenced by \cite{sigspatial1}. We first use spatial hashing (via a kd-tree) on endpoints and extrema of the curves to select a good set of candidate curves, and then use our decider from Section \ref{sec:decider} to decide those candidates.

\paragraph{Benchmark.} We build a query benchmark similar to the one used in \cite{sigspatial1}. For each $k \in \{0, 1, 10, 100, 1000\}$, we select a random curve $\pi \in \mathcal{D}$ and then pick a threshold distance $\delta$ such that a query of the form $(\pi, \delta)$ returns exactly $k+1$ curves (note that the curve $\pi$ itself is also always returned). We repeat this 1000 times for each value of $k$ and also create such a benchmark for each of the three data sets.

\paragraph{Running Times.} We compare our implementation with the top three implementations of the \giscup on this benchmark. The results are shown in Table \ref{tab:sigspatial_comparison}. Again the running time improvement of our implementation depends on the data set. For \characters the maximal improvement factor over the second best implementation is $14.6$, for \sigspatial $17.3$, and for \geolife $29.1$. For \sigspatial and \characters it is attained at $k=1000$, while for \geolife it is reached at $k=100$ but $k=1000$ shows a very similar but slightly smaller factor.

To give deeper insights about the single parts of our decider, a detailed analysis of the running times of the single parts of the algorithm is shown in Table \ref{tab:times_parts}. Again we witness different behavior depending on the data set. It is remarkable that for \sigspatial the running time for $k=1000$ is dominated by the greedy filter. This suggests that improving the filters might still lead to a significant speed-up in this case. However, for most of the remaining cases the running time is clearly dominated by the complete decider, suggesting that our efforts of improving the state of the art focused on the right part of the algorithm.

\begin{landscape}

\begin{table}
% \tiny
% \makebox[\textwidth][c]{
\centering
\begin{tabular}{|l|rrrrr|rrrrr|rrrrr|}
\hline
& \multicolumn{5}{c|}{\sigspatial} & \multicolumn{5}{c|}{\characters} & \multicolumn{5}{c|}{\geolife} \\
\hline
$k$		& 0 & 1 & 10 & 100 & 1000 & 0 & 1 & 10 & 100 & 1000 & 0 & 1 & 10 & 100 & 1000 \\
\hline
\cite{sigspatial1} & 0.094 & 0.123 & 0.322 & 1.812 & 8.408 & 0.187 & 0.217 & 0.421 & 2.222 & 17.169 & 0.298 & 0.741 & 4.327 & 33.034 & 109.44 \\
\cite{sigspatial2} & 0.421 & 0.618 & 1.711 & 7.86 & 35.704 & 0.176 & 0.28 & 0.611 & 3.039 & 17.681 & 3.627 & 6.067 & 26.343 & 120.509 & 415.548 \\
\cite{sigspatial3} & 0.197 & 0.188 & 0.643 & 5.564 & 76.144 & 0.142 & 0.147 & 0.222 & 1.849 & 22.499 & 2.614 & 4.112 & 16.428 & 166.206 & 1352.19 \\
ours & 0.017 & 0.007 & 0.026 & 0.130 & 0.490 & 0.004 & 0.020 & 0.058 & 0.301 & 1.176 & 0.027 & 0.089 & 0.341 & 1.108 & 3.642 \\
\hline
\end{tabular}
% }
\caption{Comparing the running times (in $s$) of the queries of the top three implementations of the \giscup 2017 with our new implementation on the query benchmark on all data sets (1000 queries per entry).}
\label{tab:sigspatial_comparison}
\end{table}

\begin{table}
% \tiny
% \makebox[\textwidth][c]{
\centering
\begin{tabular}{|l|rrrrr|rrrrr|rrrrr|}
\hline
& \multicolumn{5}{c|}{\sigspatial} & \multicolumn{5}{c|}{\characters} & \multicolumn{5}{c|}{\geolife} \\
\hline
$k$		& 0 & 1 & 10 & 100 & 1000 & 0 & 1 & 10 & 100 & 1000 & 0 & 1 & 10 & 100 & 1000 \\
\hline
spatial hashing & 0.002 & 0.003 & 0.005 & 0.017 & 0.074 & 0.002 & 0.002 & 0.004 & 0.011 & 0.032 & 0.006 & 0.009 & 0.016 & 0.032 & 0.091 \\
greedy filter & 0.004 & 0.006 & 0.024 & 0.143 & 0.903 & 0.004 & 0.010 & 0.032 & 0.153 & 0.721 & 0.009 & 0.017 & 0.060 & 0.273 & 1.410 \\
adaptive equal-time filter & 0.000 & 0.001 & 0.006 & 0.030 & 0.088 & 0.001 & 0.004 & 0.018 & 0.088 & 0.424 & 0.005 & 0.017 & 0.063 & 0.273 & 1.211 \\
negative filter & 0.001 & 0.002 & 0.010 & 0.044 & 0.107 & 0.003 & 0.012 & 0.038 & 0.152 & 0.309 & 0.008 & 0.020 & 0.069 & 0.200 & 0.606 \\
complete decider & 0.002 & 0.011 & 0.044 & 0.214 & 0.330 & 0.005 & 0.030 & 0.109 & 0.671 & 2.639 & 0.062 & 0.210 & 0.998 & 3.025 & 8.760 \\
\hline
\end{tabular}
% }
\caption{Timings (in $s$) of the single parts of our query algorithm on the query benchmark on all three data sets. To avoid confusion, note that the sum of the times in this table do not match the entries in Table \ref{tab:sigspatial_comparison} as those are parallelized timings and additionally the timing itself introduces some overhead.}
\label{tab:times_parts}
\end{table}
\end{landscape}

\subsection{Other Experiments}

The main goal of the complete decider was to reduce the number of recursive calls that we need to consider during the computation of the free-space diagram. Due to our optimized algorithm to compute simple boundaries with adaptive step size, we expect roughly a constant (or possibly polylogarithmic) running time effort per box, essentially independent of the size of the box. To test this hypothesis, we ask whether the number of recursive calls is indeed correlated with the running time. To test this, we measured the time for each complete decider call in the query benchmark and plotted it over the number of boxes that were considered in this call. The result of this experiment is shown in Figure \ref{fig:boxes_vs_time}. We can see a practically (near-)linear correlation between the number of boxes and the running time.

\begin{figure}
	\centering
	\includegraphics[width=.7\textwidth]{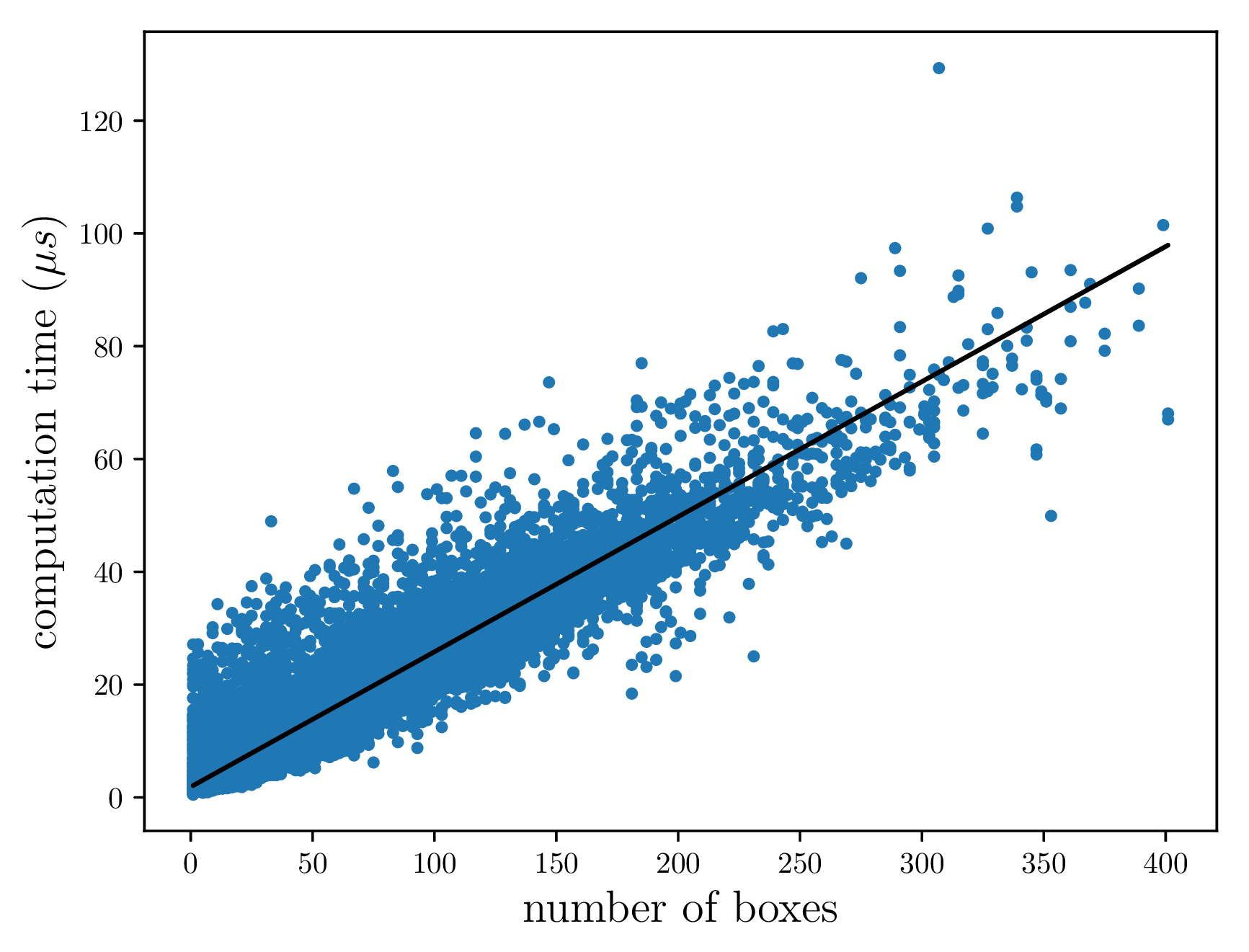}
	\caption{Shows how much time a call to the complete decider takes plotted over the number of boxes that the free-space diagram creates in total (\ie, even if a box is later split, it is still counted). The data are all exact computations (\ie, those where neither kd-tree nor filter decided) issued for the \sigspatial query benchmark. The black line is the linear regression ($r^2 = 0.91$).}
	\label{fig:boxes_vs_time}
\end{figure}

%% file: trunk/certificates.tex
\newcommand{\lowerRight}{\mathrm{lowerRight}}
\newcommand{\upperLeft}{\mathrm{upperLeft}}

\section{Certificates}
\label{sec:certificates}

Whenever we replace a naive implementation in favor of a fast, optimized, but typically more complex implementation, it is almost unavoidable to introduce bugs to the code. As a useful countermeasure the concept of \emph{certifying algorithms} has been introduced; we refer to~\cite{McConnellMNS11} for a survey. In a nutshell, we aim for an implementation that outputs, apart from the desired result, also a proof of correctness of the result. Its essential property is that the certificate should be \emph{simple} to check (i.e., much simpler than solving the original problem). In this way, the certificate gives any user of the implementation a simple means to check the output for any conceivable instance.

Following this philosophy, we have made our implementation of the Fréchet decider \emph{certifying}: for any input curves $\pi, \sigma$ and query distance $\delta$, we are able to return, apart from the output whether the Fréchet distance of $\pi$ and $\sigma$ is at most $\delta$, also a certificate $c$. On our realistic benchmarks, constructing this certificate slows down the Fréchet decider by roughly 50\%. The certificate $c$ can be checked by a simple verification procedure consisting of roughly 200 lines of code.

In Sections~\ref{sec:yescertificates} and~\ref{sec:nocertificates}, we define our notion of YES and NO certificates, prove that they indeed certify YES and NO instances and discuss how our implementation finds them. In Section~\ref{sec:certificatechecker}, we describe the simple checking procedure for our certificates. Finally, we conclude with an experimental evaluation in Section~\ref{sec:certificateExperiments}.

\subsection{Certificate for YES Instances}
\label{sec:yescertificates}
To verify that $\dF(\pi, \sigma) \le \delta$, by definition it suffices to give a feasible traversal, i.e., monotone and continuous functions $f: [0,1] \to [1,n]$ and $g: [0,1] \to [1,m]$ such that for all $t\in [0,1]$, we have $(\pi_{f(t)}, \sigma_{g(t)}) \in F$, where $F = \{(p,q) \in [1,n]\times[1,m] \mid \|\pi_p - \sigma_q\| \le \delta\}$ denotes the free-space (see Section~\ref{subsec:freespace_diagram}). We slightly simplify this condition by discretizing $(f(t), g(t))_{t\in [0,1]}$, as follows. 

\begin{definition}\label{def:yescert}
We call $T = (t_1, \dots, t_\ell)$ with $t_i \in [1,n]\times[1,m]$ a YES certificate if it satisfies the following conditions: (See also Figure~\ref{fig:YEScertificate} for an example.)

\begin{enumerate}
\item (start) $t_1 = (1, 1) \in F$,
\item (end) $t_\ell = (n,m) \in F$,
\item (step) For any $t_k = (p,q)$ and $t_{k+1} = (p',q')$, we have either
\begin{enumerate}
\item $p'=p$ and $q' > q$: In this case, we require that $(p, \bar{q})\in F$ for all $\bar{q} \in \{q, \lceil q \rceil, \dots, \lfloor q'\rfloor, q'\}$,\label{enum:vertical} 
\item $q'=q$ and $p' > p$: In this case, we require that $(\bar{p}, q)\in F$ for all $\bar{p} \in \{p, \lceil p \rceil, \dots, \lfloor p'\rfloor, p'\}$,\label{enum:horizontal}
\item $i \le p < p' \le i+1$, $j \le q < q' \le j+1$ for some $i\in \{1,\dots, n\}, j \in \{1, \dots, m\}$: In this case, we require that $(p, q), (p', q') \in F$. \label{enum:diagonal}
\end{enumerate}
\end{enumerate}
\end{definition}

It is straightforward to show that a YES certificate $T$ proves correctness for YES instances as follows. 
\begin{proposition}
Any YES certificate $T = (t_1, \dots, t_\ell)$ with $t_i \in [1,n]\times[1,m]$ proves that $\dF(\pi, \sigma) \le \delta$.
\end{proposition}
\begin{proof}
View $T$ as a polygonal curve in $[1,n]\times[1,m]$ and let $\tau : [0,1] \to [1,n]\times[1,m]$ be a reparameterization of $T$. Let $f,g$ be the projection of $\tau$ to the first and second coordinate, respectively. Note that by the assumption on $T$, $f$ and $g$ are monotone and satisfy $(f(0),g(0)) = (1,1)$ and $(f(1), g(1)) = (n,m)$. We claim that $(f(t),g(t)) \in F$ for all $t\in [0,1]$, which thus yields $\dF(\pi, \sigma) \le \delta$ by definition. 

To see the claim, we recall that for any \emph{cell} $[i,i+1]\times[j,j+1]$, the free-space restricted to this cell, i.e., $F\cap [i,i+1]\times [j,j+1]$, is convex (as it is the intersection of an ellipse with $[i,i+1] \times [j,j+1]$, see~\cite{alt_95}). Observe that for any segment from $t_k=(p,q)$ to $t_{k+1}=(p',q')$, we (implicitly) decompose it into subsegments contained in single cells (e.g., for $p'=p$ and $q' > q$, the segment from $(p,q)$ to $(p,q')$ is decomposed into the segments connecting the sequence $(p,q), (\lceil p \rceil, q), \dots, (\lfloor p' \rfloor, q), (p'q')$. As each such subsegment is contained in a single cell, by convexity we see that the whole subsegment is contained in $F$ if the corresponding endpoints of the subsegment are in $F$. This concludes the proof.
\end{proof}

It is not hard to prove that for YES instances, such a certificate always exists (in fact, there always is a certificate of length $O(n+m)$). Furthermore, for each YES instance in our benchmark set, our implementation indeed finds and returns a YES instance, in a way we describe next.

\begin{figure}
\centering
	\includegraphics[width=0.49\textwidth]{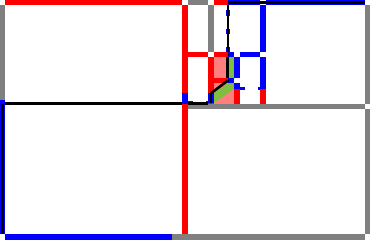}
	\includegraphics[width=0.48\textwidth]{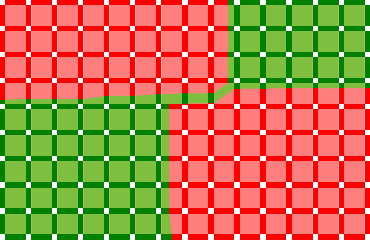}
	\caption{Example of a YES instance and its certificate. The right picture shows the free-space of the instance. The left picture illustrates the parts of the free-space explored by our algorithm and indicates the computed YES certificate by black lines.}
	\label{fig:YEScertificate}
\end{figure}

\paragraph{Certifying positive filters.}
It is straightforward to construct YES certificates for instances that are resolved by our positive filters (Bounding Box Check, Greedy and Adaptive Equal-Time): All of these filters implicitly construct a feasible traversal. In particular, for any instance for which the Bounding Box Check applies (which shows that any pair of points of $\pi$ and $\sigma$ are within distance $\delta$), already the sequence $((1,1), (n,1), (n,m))$ yields a YES certificate.
																					  
For Greedy, note that the sequence of positions $(i,j)$ visited in Algorithm~\ref{alg:greedy_filter} yields a YES certificate: Indeed, any step from $(i,j)$ is either a vertical step to $(i, j+s)$ (corresponding to case \ref{enum:vertical}),  a horizontal step to $(i+s, j)$ (corresponding to case \ref{enum:horizontal}), or a diagonal step \emph{within a cell} to $(i+1, j+1)$ (corresponding to Case~\ref{enum:diagonal} of Definition~\ref{def:yescert}). Furthermore, such a step is only performed if it stays within the free-space.

Finally, for Adaptive Equal-Time, we also record the sequence of positions $(i,j)$ visited in Algorithm~\ref{alg:greedy_filter} (recall that here, we change the set of possible steps for $s > 1$ to $S=\{(i+s, j+ s')\}$ with $s'=\lfloor \frac{m-j}{n-i} \cdot s \rfloor$) -- with the only difference that we need to replace any step from $(i,j)$  to $(i+s,j+s')$ by the sequence $(i,j), (i+s,j), (i+s,j+s')$. Note that this sequence satisfies Condition (step) of Definition~\ref{def:yescert}, as Adaptive Equal-Time only performs this step if it can verify that \emph{all pairwise distances} between $\pi_{i\dots i+s}$ and $\sigma_{j\dots j+s'}$ are bounded by $\delta$.

\paragraph{Certifying YES instances in the complete decider.}
Recall that the complete decider via free-space exploration decides an instance by recursively determining, given the inputs $B_l^R, B_b^R$ of a box $B$, the corresponding outputs $B_r^R, B_t^R$. In particular, YES instances are those with $(n,m) \in B_t^R$ (or equivalently $(n,m) \in B_r^R$) for the box $B=[1,n]\times [1,m]$. To certify such instances, we memorize for each point in $B_r^R$  and $B_t^R$ a predecessor of a feasible traversal from $(1,1)$ to this point. Note that here, it suffices to memorize such a predecessor only for the \emph{first}, i.e., lowest or leftmost, \emph{point of each interval} in $B_r^R$ and $B_t^R$ (as any point in this interval can be reached by traversing to the first point of the interval and then along this reachable interval to the destination point). This gives rise to a straightforward recursive approach to determine a feasible traversal. 

In the complete decider, whenever we determine some output $B_t^R$, it is because of one of the following reasons: (1) one of of our pruning rules is successful, (2) the box $B$ is on the cell-level, or (3) we determine $B_t^R$ as the union of the outputs $(B_1)_t^R$, $(B_2)_t^R$ of the boxes $B_1,B_2$ obtained by splitting $B$ vertically. Note that we only need to consider the case in which $B_t^R$ is determined as non-empty (otherwise nothing needs to be memorized). Let us consider each case separately. 

If reason (1) determines a non-empty $B_t^R$, then this happens either by Rule \rom{3}b or by Rule \rom{3}c. Note that in both cases, $B_t^R$ consists of a single interval. If Rule \rom{3}b applies, then the last, i.e., topmost, point on $B_l$ is reachable and proves that the \emph{free prefix} of $B^R_t$ is reachable. Thus, we store \emph{the last interval of $B_l^R$} as the responsible interval for the (single) interval in $B^R_t$. Similarly, if Rule \rom{3}c applies, then consider the first, i.e., leftmost, point $(i, j_{\max})$ on $B_t$. Since the rule applies, the opposite point $(i,j_{\min})$ on $B_l$ must be reachable and the path $\{i\} \times [j_{\min},j_{\max}]$ must be free. Thus, we can store \emph{the interval of $B^R_b$ containing $(i,j_{\min})$} as the responsible interval for the (single) interval in $B^R_t$.  

If reason (2) determines a non-empty $B_t^R$, then we are on a cell-level. In this case, either $B_l^R$ or $B_b^R$ is a non-empty interval, and we can store such an interval as the responsible interval for the (single) interval in $B_t^R$. Finally, if reason (3) determines a non-empty $B_t^R$, then we simply keep track of the responsible interval for each interval in $(B_1)_t^R$ and $(B_2)_t^R$ (to be precise, if the last interval of $(B_1)^R_t$ and the first interval of $(B_2)^R_t$ overlap by the boundary point, we merge the two corresponding intervals and only keep track of the responsible interval of the last interval of $(B_1)^R_t$ and can safely forget about the responsible interval of the first interval of $(B_2)^R_t$.

Note that we proceed analogously for outputs $B_r^R$. Furthermore, the required memorization overhead is very limited.

It is straightforward to use the memorized information to compute a YES certificate recursively: Specifically, to compute a YES certificate reaching some point $x$ on an output interval $I$, we perform the following steps. Let $J$  be the responsible interval of $I$. We recursively determine a YES certificate reaching the first point $J$. Then we append a point to the certificate to traverse to the point of $J$ from which we can reach the first point of $I$ (this point is easily determined by distinguishing whether we are on the cell-level, and whether $J$ is opposite to $I$ or intersects $I$ in a corner point). We append the first point of $I$ to the certificate, and finally append the point $x$ to the certificate.\footnote{To be precise, we only append a point if it is different from the last point of the current certificate.} By construction, the corresponding traversal never leaves the free-space. Using this procedure, we can compute a YES certificate by computing a YES certificate reaching $(n,m)$ on the last interval of $B_t^R$ for the initial box $B=[1,n]\times[1,m]$.
%TODO: {Karl's comment: shorten to "using the memorized predecessors, we can easily construct a YES certificate as the chain of predecessors leading from $(n,m)$ to $(1,1)$}

\subsection{Certificate for NO Instances}
\label{sec:nocertificates}

We say that a point $(p,q)$ lies on the bottom boundary if $q = 1$, on the right boundary if $p = n$, on the top boundary if $q = m$, and on the left boundary if $p=1$. Likewise, we say that a point $(p',q')$ lies to the lower right of a point $(p, q)$, if $p \le p'$ and $q \ge q'$. 

\begin{definition}\label{def:nocert}
We call $T = (t_1, \dots, t_\ell)$ with $t_i \in [1,n]\times[1,m]$ a NO certificate if it satisfies the following conditions: (See also Figure~\ref{fig:NOcertificate} for an example.)
\begin{enumerate}
\item (start) $t_1$ lies on the right or bottom boundary and $t_1 \notin F$,
\item (end) $t_\ell$ lies on the left or upper boundary and $t_\ell \notin F$,
\item (step) For any $t_k = (p,q)$ and $t_{k+1} = (p',q')$, we have either
\begin{enumerate}
\item $p'=p$ and $q' > q$: In this case, for any neighboring elements $\bar{q}_1, \bar{q}_2$ in $q, \lceil q \rceil, \dots, \lfloor q' \rfloor , q'$, we require that $(\{p\}\times [\bar{q}_1, \bar{q}_2]) \cap F = \emptyset$,\label{enum:verticalNO} 
\item $q'=q$ and $p' < p$: In this case, for any neighboring elements $\bar{p}_1, \bar{p}_2$ in $p', \lceil p' \rceil, \dots, \lfloor p \rfloor , p$, we require that $([\bar{p}_1, \bar{p}_2]\times \{q\}) \cap F = \emptyset$,\label{enum:horizontalNO} 
\item $t_{k+1}$ lies to the lower right of $t_k$, i.e., $p \le p'$ and $q \ge q'$.  \label{enum:monotonicity}
\end{enumerate}
\end{enumerate}
\end{definition}

\begin{figure}
	\includegraphics[width=0.49\textwidth]{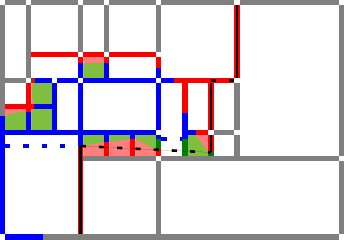}
	\includegraphics[width=0.49\textwidth]{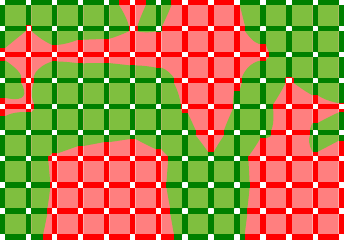}
	\caption{Example of a NO instance and its certificate. The right picture shows the free-space of the instance. The left picture illustrates the parts of the free-space explored by our algorithm and indicates the computed NO certificate by black lines.}
	\label{fig:NOcertificate}
\end{figure}

We prove that a NO certificate $T$ proofs correctness for NO instances as follows.

\begin{proposition}
Any NO certificate $T = (t_1, \dots, t_\ell)$ with $t_i \in [1,n]\times[1,m]$ proves that $\dF(\pi, \sigma) > \delta$.
\end{proposition}
\begin{proof}
We inductively prove that no feasible traversal from $(1,1)$ to $(n,m)$ can visit any point to the lower right of $t_i$, for all $1\le i \le \ell$. As an immediate consequence, $\dF(\pi, \sigma) > \delta$, since $t_\ell$ lies on the left or upper boundary and thus \emph{any} feasible traversal must visit a point to the lower right of $t_\ell$ -- hence, such a traversal cannot exists.

As base case, note that $t_1$ lies on the right or bottom boundary and is not contained in the free-space. Thus, by monotonicity, no feasible traversal can visit any point to the lower right of $t_1$. Thus, assume that the claim is true for $t_i = (p,q)$ and consider the next point $t_{i+1} = (p', q')$ in the sequence. If $t_{i+1}$ lies to the lower right of $t_{i}$, the claim is trivially fulfilled for $t_{i+1}$ by monotonicity. If, however, $p' = p$ and $q' > q$, then Condition~\ref{enum:verticalNO} of Definition~\ref{def:nocert} is equivalent to $({p}\times [q, q']) \cap F = \emptyset$.  Note that any feasible traversal visiting a point to the lower right of $t_{i+1}$ must either visit a point to the lower right of $t_{i}$  -- which is not possible by inductive assumption -- or must cross the path $\{p\}\times [q,q']$ -- which is not possible as $(\{p\}\times [q,q'])\cap F = \emptyset$. We argue symmetrically for the case that $q' = q$ and $p' < p$. This concludes the proof.
\end{proof}

Note that our definition of NO certificate essentially coincides with the definition of a cut of positive width in~\cite{BuchinOS19}. For NO instances, such a NO certificate always exists (in contrast to YES certificates, the shortest such certificate is of length $\Theta(n^2)$ in the worst case). For all NO instances in our benchmark sets, our implementation manages to find and return such a NO certificate, in a way we describe next. 

\paragraph{Certifying the negative filter.}
It is straightforward to compute a NO certificate for instances resolved by our negative filter. Note that this filter, if successful, determines an index $p\in \{1,\dots, n\}$  such that $\pi_p$ is far from all points on $\sigma$, or symmetrically an index $q\in \{1, \dots, m\}$ such that $\sigma_q$ is far from all points on $\pi$. Thus, in these cases, we can simply return the NO certificate $((p, 1), (p, m))$ or $((n,q), (1, q))$, respectively.

\paragraph{Certifying NO instances in the complete decider.}
Whenever the complete decider via free-space exploration returns a negative answer, the explored parts of the free-space diagram must be sufficient to derive a negative answer. This gives rise to the following approach: Consider all non-free segments computed by the complete decider. We start from a non-free segment touching  the bottom or right boundary and traverse non-free segments (possibly also making use of monotonicity steps according to Case~\ref{enum:monotonicity} of Definition~\ref{def:nocert}) and stop as soon as we have found a non-free segment touching the left or top boundary.

Formally, consider Algorithm~\ref{alg:NOcertificate}. Here, we use the notation that $\lowerRight(I)$ denotes the lower right endpoint of $I$, i.e., the right endpoint if $I$ is a horizontal segment and the lower endpoint if $I$ is a vertical segment. Analogously, $\upperLeft(I)$ denotes the upper left endpoint of $I$.
%TODO: describe as a graph problem?

\begin{algorithm}[t]
\begin{algorithmic}[1]
\Procedure{ComputeNOCertificate}{$\pi, \sigma, \delta$}
\State $N \gets $ non-free segments determined by \Call{CompleteDecider}{$\pi, \sigma, \delta$}
\State $Q \gets \{ I \in N \mid \lowerRight(I) \text{ lies on bottom or right boundary} \}$
\State Build orthogonal range search data structure $D$,\\ \hfill storing all $I\in N \setminus Q$ under the key $\lowerRight(I)$.
\While{$Q \ne \emptyset$}
\State Pop any element $I$ from $Q$
\If{$\upperLeft(I)$ lies on top or left boundary}
\State Reconstruct sequence of intervals leading to $I$
\State \Return corresponding NO certificate
\Else
\State $Q' \gets D.\Call{ReportAndDelete}{\upperLeft(I)}$ \\\Comment{reports $J$ if $\lowerRight(J)$ is to the lower right of $\upperLeft(I)$}
\State $Q \gets Q \cup Q'$
\EndIf
\EndWhile
\EndProcedure
\end{algorithmic}
\caption{High-level code for computing a NO certificate.}
\label{alg:NOcertificate}
\end{algorithm}

The initial set of non-free segments in Algorithm~\ref{alg:NOcertificate} consists of the non-free segments of \emph{all simple boundaries} determined by the complete decider via free-space exploration. We maintain a queue $Q$ of non-free segments, which initially contains all non-free segments touching the right or bottom boundary. Furthermore, we maintain a data structure $D$ of yet \emph{unreached} non-free intervals. Specifically, we require $D$ to store intervals $I$ under the corresponding key $\lowerRight(I) \in [1,n]\times[1,m]$ in a way to support the query $\Call{ReportAndDelete}{p}$: Such a query returns all $I \in D$ such that $\lowerRight(I)$ lies to the lower right of $p$ and deletes all returned intervals from $D$. 

Equipped with such a data structure, we can traverse all elements in the queue as follows: We delete any interval $I$ from $Q$ and check whether it reaches the upper or left boundary. If this is the case, we have (implicitly) found a NO certificate, which we then reconstruct (by memorizing why each element of the queue was put into the queue). Otherwise, we add to $Q$ all intervals from $D$ that can be reached by a monotone step (according to Case~\ref{enum:monotonicity} of Definition~\ref{def:nocert}) from $\upperLeft(I)$; these intervals are additionally deleted from $D$. %If we fail in reaching the upper or left boundary while traversing the queue, the explored information (i.e., the initial set of non-free segments) was not sufficient to determine a NO certificate, pointing to a bug in our implementation of the complete decider via freespace exploration.

To implement $D$, we observe that it essentially asks for a 2-dimensional orthogonal range search data structure where the ranges are unbounded in two directions (and bounded in the other two). Already for the case of 2-dimensional ranges with only a single unbounded direction (sometimes called 1.5-dimensional), a very efficient solution is provided by a classic data structure due to McCreight, the priority search tree data structure~\cite{McCreight85}. We can adapt it in a straightforward manner to implement $D$ such that it (1) takes time $O(d \log d)$ and space $O(d)$ to construct $D$ on an initial set of size $d$ and (2) supports $\Call{ReportAndDelete}{p}$ queries in time $O(k + \log d)$, where $k$ denotes the number of reported elements. Thus, Algorithm~\ref{alg:NOcertificate} can be implemented to run in time $O(|N| \log |N|)$.

{
\newcommand\setrow[1]{\gdef\rowmac{#1}#1\ignorespaces}
\newcommand\clearrow{\global\let\rowmac\relax}
\clearrow
\newcommand{\lvl}{~~~}
\begin{landscape}
\begin{table}
\vspace{-1.5cm}
%\begin{tabular}{|>{\rowmac}l|>{\rowmac}r>{\rowmac}r>{\rowmac}r>{\rowmac}r>{\rowmac}r|>{\rowmac}r>{\rowmac}r>{\rowmac}r>{\rowmac}r>{\rowmac}r|>{\rowmac}r>{\rowmac}r>{\rowmac}r>{\rowmac}r>{\rowmac}r<{\clearrow}|}
\centering
\begin{tabular}{|>{\rowmac}l|>{\rowmac}r>{\rowmac}r>{\rowmac}r>{\rowmac}r>{\rowmac}r<{\clearrow}|}
\hline
& \multicolumn{5}{c|}{\sigspatial} \\
\hline
$k$		& 0 & 1 & 10 & 100 & 1000 \\
\hline
\setrow{\bfseries} computation without certification & 6.9 & 21.3 & 84.5 & 429.4 & 1409.1\\  
\setrow{\bfseries} certifying computation & 10.0 & 29.6 & 117.8 & 553.8 & 1840.2\\
\lvl --computation of certificates & 1.0 & 3.7 & 12.0 & 40.9 & 65.7\\
\lvl \lvl --YES certificates (complete decider) & 0.0 & 0.4 & 1.6 & 8.2 & 12.1\\
\lvl \lvl --NO certificates (complete decider) & 1.0 & 3.3 & 10.0 & 31.4 & 50.0\\
\setrow{\bfseries} checking certificates & 6.4 & 13.9 & 63.7 & 426.5 & 3803.2\\
\lvl --checking filter certificates & 6.0 & 11.3 & 51.6 & 361.2 & 3666.7\\
\lvl --checking complete decider certificates & 0.4 & 2.6 & 12.1 & 65.3 & 136.5\\
\hline
\end{tabular}

\vspace{1em}

\begin{tabular}{|>{\rowmac}l|>{\rowmac}r>{\rowmac}r>{\rowmac}r>{\rowmac}r>{\rowmac}r<{\clearrow}|}
\hline
& \multicolumn{5}{c|}{\characters} \\
\hline
$k$		& 0 & 1 & 10 & 100 & 1000 \\
\hline
\setrow{\bfseries} computation without certification &  12.1 & 55.5 & 205.8 & 1052.8 & 4080.3\\
\setrow{\bfseries} certifying computation & 20.3 & 91.6 & 311.2 & 1589.8 & 5895.8\\
\lvl --computation of certificates & 4.0 & 20.0 & 59.6 & 220.7 & 470.2\\
\lvl \lvl --YES certificates (complete decider) & 0.0 & 0.5 & 2.3 & 24.9 & 181.3\\
\lvl \lvl --NO certificates (complete decider) & 3.9 & 19.2 & 56.3 & 186.2 & 259.4\\
\setrow{\bfseries} checking certificates & 6.3 & 21.6 & 76.8 & 457.7 & 2626.1\\
\lvl --checking filter certificates & 5.3 & 14.8 & 49.7 & 278.0 & 1759.8\\
\lvl --checking complete decider certificates & 1.0 & 6.8 & 27.2 & 179.6 & 866.3\\
\hline
\end{tabular}

\vspace{1em}

\begin{tabular}{|>{\rowmac}l|>{\rowmac}r>{\rowmac}r>{\rowmac}r>{\rowmac}r>{\rowmac}r<{\clearrow}|}
\hline
& \multicolumn{5}{c|}{\geolife} \\
\hline
$k$		& 0 & 1 & 10 & 100 & 1000 \\
\hline
\setrow{\bfseries} computation without certification & 82.2 & 251.1 & 1156.6 & 3663.1 & 11452.4\\
\setrow{\bfseries} certifying computation & 142.1 & 414.6 & 1834.4 & 5304.2 & 16248.7\\
\lvl --computation of certificates & 40.1 & 100.7 & 388.0 & 767.7 & 1827.8\\
\lvl \lvl --YES certificates (complete decider) & 0.0 & 3.2 & 20.0 & 87.6 & 247.5\\
\lvl \lvl --NO certificates (complete decider) & 39.7 & 96.7 & 364.5 & 664.8 & 1517.1\\
\setrow{\bfseries} checking certificates & 70.9 & 185.2 & 733.7 & 3595.4 & 20188.4\\
\lvl --checking filter certificates & 45.2 & 85.9 & 283.8 & 1754.0 & 12987.5\\
\lvl --checking complete decider certificates & 25.7 & 99.4 & 450.0 & 1841.4 & 7200.9\\
\hline
\end{tabular}
\caption{Certificate computation and check times on query setting benchmark (in ms). The first and second bold lines show the running time of our implementation compiled without and with certification, respectively. For the certifying variant, we also give the times to compute YES and NO certificates of the complete decider (note that filter certificates are computed on the fly by the filters and hence cannot be separately measured). Finally, we give running times for checking correctness of certificates.}
\label{tab:cert_experiments}
\end{table}
\end{landscape}
}

\subsection{Certificate Checker}
\label{sec:certificatechecker}

It remains to describe how to check the correctness of a given certificate $T = (t_1,\dots, t_\ell)$. For this, we simply verify that all properties of Definition~\ref{def:yescert} or Definition~\ref{def:nocert} are satisfied.

\paragraph{Checking YES certificates.}
Observe that the only conditions in the definition of YES instances are either simple comparisons of neighboring elements $t_k, t_{k+1}$ in the sequence or \emph{freeness tests}, specifically, whether a give position $p \in [1,n]\times [1,m]$ is free, i.e, whether $\pi_{p_1}$ and $\sigma_{p_2}$ have distance at most $\delta$. The latter test only requires interpolation along a curve segment (to obtain $\pi_{p_1}$ and $\sigma_{p_2}$) and a Euclidean distance computation. Thus, YES certificates are extremely simple to check. 

\paragraph{Checking NO certificates.}
Checking NO certificates involves a slightly more complicated geometric primitive than the freeness tests of YES certificates. Apart from simple comparisons of neighboring elements $t_{k}, t_{k+1}$, the conditions in the definition involve the following \emph{non-freeness tests}: Given a (sub)segment $\pi_{p..p'}$ with $i \le p \le p' \le i+1$ for some $i\in [n]$, as well as a point $\sigma_q$ with $q\in [1,m]$, determine whether all points on $\pi_{p..p'}$ have distance strictly larger than $\delta$ from $\sigma_q$. Besides the (simple) interpolation along a line segment to obtain $\sigma_q$, we need to determine intersection points of the line containing $\pi_{p..p'}$ and the circle of radius $\delta$ around $\sigma_q$ (if these exists). From these intersection points, we verify that $\pi_{p..p'}$ and the circle do not intersect, concluding the check.

\paragraph{}
In summary, certificate checkers are straightforward and simple to implement. 

\subsection{Certification Experiments}
\label{sec:certificateExperiments}

We evaluate the overhead introduced by computing certificates using our benchmark sets for the query setting. In particular, as our implementation can be compiled both as a certifying and a non-certifying version, we compare the running times of both versions. The results are depicted in Table~\ref{tab:cert_experiments}. Notably, the slowdown factor introduced by computing certificates ranges between 1.29 and 1.46 (\sigspatial), 1.44 and 1.67 (\characters) and 1.42 and 1.73 (\geolife). As expected, the certificate computation time is dominated by the task of generating NO certificates (which is more complex than computing YES certificates), even for large values of $k$ for which most unfiltered instances are YES instances.

At first sight, it might be surprising that checking the certificates takes longer than computing them. However, this is due to the fact that our filters often display sublinear running time behavior (by using the heuristic checks and adaptive step sizes). However, to keep our certificate checker elementary, we have not introduced any such improvements to the checker, which thus has to traverse essentially all points on the curves. This effect is particularly prominent for large values of $k$. 

%\begin{table}
%%\begin{tabular}{|>{\rowmac}l|>{\rowmac}r>{\rowmac}r>{\rowmac}r>{\rowmac}r>{\rowmac}r|>{\rowmac}r>{\rowmac}r>{\rowmac}r>{\rowmac}r>{\rowmac}r|>{\rowmac}r>{\rowmac}r>{\rowmac}r>{\rowmac}r>{\rowmac}r<{\clearrow}|}
%\begin{tabular}{|>{\rowmac}l|>{\rowmac}r>{\rowmac}r>{\rowmac}r>{\rowmac}r>{\rowmac}r<{\clearrow}|}
%\hline
%& \multicolumn{5}{c|}{\sigspatial} \\
%\hline
%$k$		& 0 & 1 & 10 & 100 & 1000 \\
%\hline
%\setrow{\bfseries} computation without certification & 
%\setrow{\bfseries} certifying computation &
%\lvl -computation of certificates &
%\lvl \lvl --YES certificates (complete decider) &
%\lvl \lvl --NO certificates (complete decider) &
%\setrow{\bfseries} checking certificates &
%\lvl -checking filter certificates &
%\lvl -checking complete decider certificates &
%\hline
%\end{tabular}
%\caption{\MARVIN{TODO!}}
%\end{table}

%% file: trunk/conclusion.tex
\section{Conclusion}
In this work we presented an implementation for computing the Fréchet distance which beats the state-of-the-art by one to two orders of magnitude in running time in the query as well as the decider setting. Furthermore, it can be used to compute certificates of correctness with little overhead. To facilitate future research, we created two benchmarks on several data sets -- one for each setting -- such that comparisons can easily be conducted. Given the variety of applications of the Fréchet distance, we believe that this result will also be of broader interest and implies significant speed-ups for other computational problems in practice.

This enables a wide range of future work. An obvious direction to continue research is to take it back to theory and show that our pruning approach provably has subquadratic runtime on a natural class of realistic curves. On the other hand, one could try to find further pruning rules or replace the divide-and-conquer approach by some more sophisticated search. To make full use of the work presented here, it would make sense to incorporate this algorithm in software libraries. Currently, we are not aware of any library with a non-naive implementation of a Fréchet distance decider or query. Finally, another possible research direction would be to work on efficient implementations for similar problems like the Fréchet distance under translation, rotation or variants of map matching with respect to the Fréchet distance. In summary, this paper should lay ground to a variety of improvements for practical aspects of curve similarity.